\documentclass[
  aps,
  prl,
  reprint,
  superscriptaddress,
  nofootinbib,
  longbibliography,
  floatfix
]{revtex4-2}

\usepackage{newtxtext,newtxmath}

\usepackage{graphicx}
\usepackage{dcolumn}
\usepackage{bm}
\usepackage{hyperref}

\usepackage{amsmath}
\usepackage{xcolor}
\usepackage{multirow,booktabs}

\usepackage{caption}
\usepackage{ragged2e}  
\usepackage{subcaption}
\usepackage{comment}

\newcounter{SI}
\newcommand*\SI%
{%
\ifnum\theSI>0%
{SI}\fi%
\ifnum\theSI=0%
{Supplementary Information (SI) \cite{supp} \setcounter{SI}{1}}\fi%
}

\begin{document}

\preprint{APS/123-QED}

\title{Fundamental Scaling Constraints for Equilibrium Molecular Computing}

\author{Erin Crawley$^{\S,}$}
\email{erincrawley@g.harvard.edu}
\affiliation{School of Engineering and Applied Sciences, Harvard University, Cambridge, Massachusetts 02138, USA} 
\affiliation{Department of Physics, Harvard University, Cambridge, Massachusetts 02138, USA}

\author{Qian-Ze Zhu$^{\S,}$}
\email{qianzezhu@g.harvard.edu}
\affiliation{School of Engineering and Applied Sciences, Harvard University, Cambridge, Massachusetts 02138, USA} 

\author{Michael P. Brenner}
\email{brenner@seas.harvard.edu}
\affiliation{School of Engineering and Applied Sciences, Harvard University, Cambridge, Massachusetts 02138, USA} 
\affiliation{Department of Physics, Harvard University, Cambridge, Massachusetts 02138, USA}

\date{\today}

\begin{abstract}
\begin{description}
\item[\S] These authors contributed equally to this work.
\end{description}

Molecular computing promises massive parallelization to explore solution spaces, but so far practical implementations remain limited due to off-target binding and exponential proliferation of competing structures. Here, we investigate the theoretical limits of equilibrium self-assembly systems for solving computing problems, focusing on the directed Hamiltonian Path Problem (HPP) as a benchmark for NP-complete problems. The HPP is encoded via particles with directional lock-key patches, where self-assembled chains form candidate solution paths. We determine constraints on the required energy gap between on-target and off-target binding for the HPP to be encoded and solved. We simultaneously examine whether components with the required energy gap can be designed. Combining these results yields a phase diagram identifying regions where HPP instances are both encodable and solvable. These results establish fundamental upper bounds on equilibrium molecular computation and highlight the necessity of non-equilibrium approaches for scalable molecular computing architectures.

\end{abstract}

\maketitle

Frontiers of computing are expanding to tackle larger problems and broader application domains. Beyond the remarkable success of silicon-based computing, there is interest in how alternative physical systems can provide novel approaches to computation and information processing. While quantum computing has made tremendous progress in recent years, achieving improvements over classical approaches \cite{arute_quantum_2019} and major advances in quantum error correction \cite{google_quantum_ai_suppressing_2023, bluvstein_logical_2024}, there has been comparatively less progress in molecular computing. However, molecular computing has innate advantages: molecular components – such as self-assembling particles or interacting polymers – can collectively explore vast configuration spaces in parallel \cite{reif1998paradigms, garzon2002biomolecular, winfree2004dna, ezziane2005dna}. At the same time, biology itself offers an existence proof for molecular computers: biological systems demonstrate that collections of molecules paired with directed energy sources can solve complex computational problems. 

Here, we consider the Directed Hamiltonian Path Problem (HPP) as a benchmark problem for molecular computing. The HPP is a canonical NP-complete problem that asks whether there exists a path in a directed graph that starts and ends at designated nodes and visits all nodes exactly once \cite{garey2002computers}. It is a well-posed problem where solving large instances is challenging for conventional silicon-based hardware. 

Leveraging molecular computers for NP-complete problems is an old idea; Adleman demonstrated that programmed DNA hybridization reactions could solve the HPP for a 7-node graph \cite{adleman1994molecular}. However, Adleman's idea has since stagnated beyond proof-of-principle demonstrations, owing to key scaling limitations \cite{ezziane2005dna, kurtz1996active}. As the problem size ($N$) increases, the solution space expands combinatorially -- yielding an exponential number of competing structures -- while suppressing crosstalk and maintaining the requisite binding specificity becomes prohibitively difficult \cite{kurtz1996active, murugan2015undesired, huntley_infocap}. This fatally hinders reliable computations as the problem size increases ($N\to\infty$).  Most work extending Adleman's initial ideas focuses on computation with biological molecules, including enzymes and proteins, which likewise suffer from poor binding specificity \cite{smith11996dna,roweis1998sticker, faulhammer2000molecular, hug2001strategies,bar2002protein}. 

Yet, the key question is not whether {\sl existing} biological molecules can solve difficult computational problems. Rather, it is whether there exists {\sl any} molecular design for the underlying molecules that simultaneously circumvents the exponential growth of competing structures and avoids the crosstalk catastrophe that limits the computational ability of molecular systems. 

Recent convergence of capabilities in synthetic biology, protein engineering, and nanoscale manufacturing now provide practical technologies to enable more refined approaches.
In principle, such a design could operate in or out of equilibrium, with the only constraint being that it can be designed with existing components. As an example, one may modify the information carrying model with the binding characteristics of components \cite{rogers2015programming,mcmullen2021dna,king2024programming}. Alternatively, the binding characteristics could change depending on the local environment of a bound particle \cite{zhu2024proofreading, evans2024designing, metson2025designing}.

Here, we consider a subset of this general problem and ask whether it is possible to design molecular components which can solve the HPP when the system is constrained to be in equilibrium. We focus on two competing requirements: (1) that a molecular system can encode and solve the target problem, and (2) that the necessary interactions are designable with realistic components. These simultaneous constraints are in tension, and we demonstrate that equilibrium assembly is only possible in a small sliver of the designable phase space.

In contrast to Adleman's pure DNA strand encoding of this problem, we encode the HPP as a self-assembling system of patchy particles: each edge in the directed graph corresponds to a unique particle species, and each particle has directional interacting patches corresponding to the source and target nodes of that edge (Fig.~\ref{fig:1}(a,b)). Two particle species are designed to interact strongly if the target node of one matches the source node of the other, and weakly otherwise to model crosstalk, as shown in an example interaction matrix in Fig.~\ref{fig:1}(c). 

Self-assembled linear chains of interacting ``edge'' particles represent candidate paths through the graph, including the possibility of node mismatch due to crosstalk. The Hamiltonian path for a graph with $N$ nodes, if it exists, corresponds to a specific self-assembled sequence of $N-1$ particles with all strong interactions. In this paper, we restrict to \textit{acyclic} graphs which contain no loops; results and discussion for general graphs (including cycles) are provided in the \SI.

We analyze the system in the grand canonical ensemble, where each species of edge particle $i$ has an associated concentration $c_i$. Each pair of particles $i, j$ interacts with Boltzmann weight $e^{-\beta E_{ij}}$, where $E_{ij}$ is the interaction energy and $\beta=1/(k_B T)$. The interaction energy matrix then takes the form
\begin{equation} \label{eq:Eij}
    E_{ij} = \begin{cases}
        s, \ \  \text{target node of } i \text{ equals source node of } j \\
        w, \ \ \text{otherwise}
    \end{cases}
\end{equation}
where $s<0$ is the strong binding energy and $w<0$ is the average off-target binding energy, satisfying $w > s$. The yield of the Hamiltonian path is given by the equilibrium fraction of the Hamiltonian path partition function as $Y_{\mathrm{HP}} = {Z_{\mathrm{HP}}}/{Z_{\text{all}}}$, where $Z_{\text{all}}$ is the total partition function summing over all possible structures.  An illustration of the yield of the Hamiltonian path is shown in Fig.~\ref{fig:1}(d).
Following \cite{murugan2015undesired}, $Z_{all}$ can be exactly calculated using a transfer matrix method, allowing us express the yield as 
\begin{equation}
    \label{eq:HP_yield}
    Y_\text{HP} = \frac{Z_\text{HP}}{Z_\text{all}}=\frac{\left(\prod_{i \in \text{HP}} c_i\right) e^{-(N-2) \beta s}}{\sum_{i,j}c_i\left[\left( \mathbb{I}- \mathbf{B}\right)^{-1}\right]_{ij}},
\end{equation}
where $\mathbf{B}$ is a matrix with entries $B_{ij} = e^{-\beta E_{ij}} c_j$. Given this parameterization of the yield, we can optimize over component concentrations to maximize the yield.

Before optimizing for concentrations, we note necessary constraints on the energy gap, $w-s$. For the yield to be maintained, the energetic contributions must outweigh the entropic ones for the desired structure. 
For the trivial case with no extraneous edges outside the Hamiltonian path, the analysis of heterogeneous 1D chain self-assembly in \cite{murugan2015undesired} directly yields the bound: $w-s > 2\log(N)k_B T$. 
As such, we parameterize the energy gap with two parameters, $a$ and $b$, as
\begin{equation}
w - s = a \log(N) + b.   
\label{eq:e_params}
\end{equation}
In \eqref{eq:e_params}, $a \geq 2$ is the energy gap-scaling coefficient which determines the strength of the energy gap $w-s$, in units of $k_B T/\ln(N)$, and $b\geq 0$ is an offset parameter. We will demonstrate below that solving general HPP instances will require even stronger constraints on the gap-scaling coefficient $a$.

\begin{figure}
    \centering
    \includegraphics[width = 0.9\linewidth]{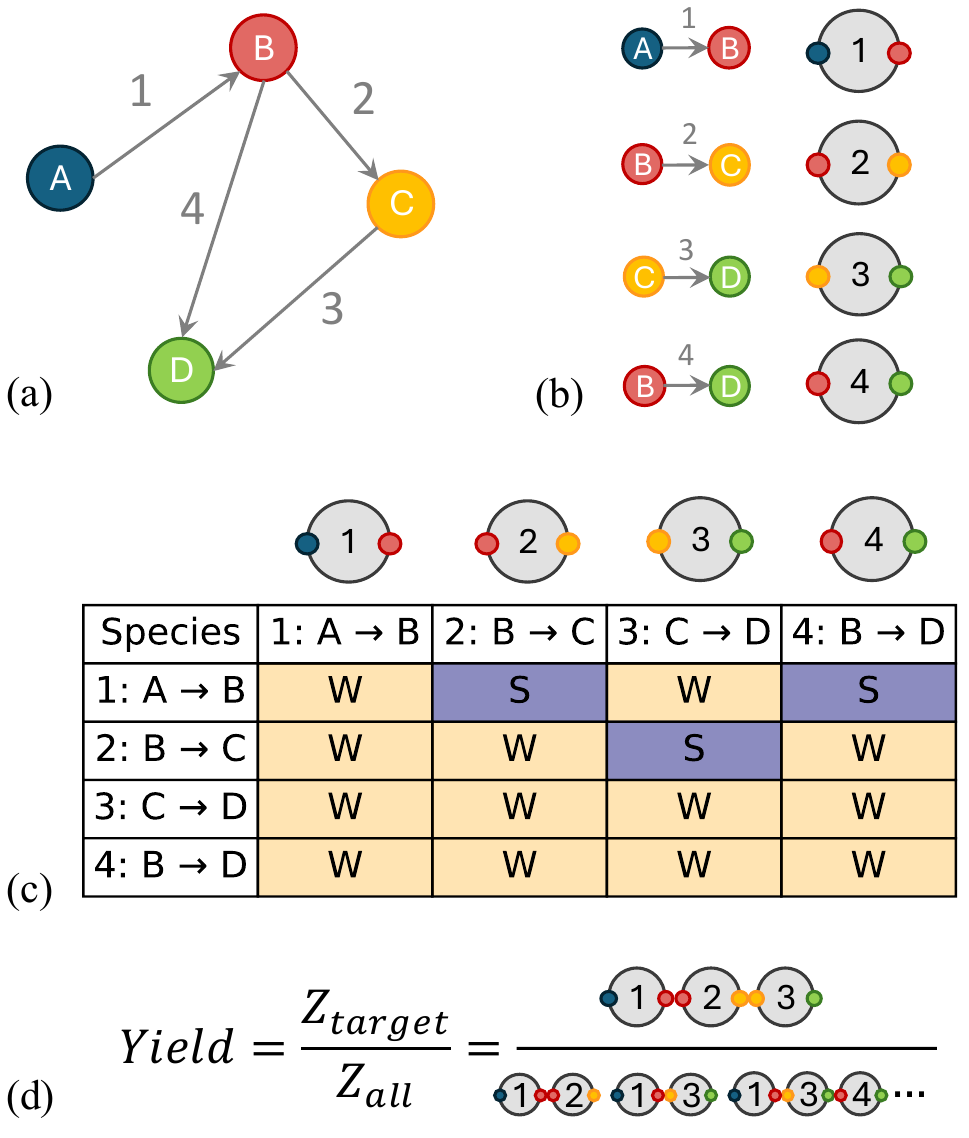}
    \caption{\justifying \textbf{Encoding the Hamiltonian Path Problem into a patchy particle system.} (a) A directed graph with four nodes and four directed edges, including one Hamiltonian path $A\rightarrow B\rightarrow C \rightarrow D$ (edges $1\rightarrow2\rightarrow3$). (b) Each directed edge is represented as a unique particle species with two patches: the head and tail nodes. (c) The interaction matrix encoding pairwise binding strengths between species. A strong interaction $s$ is assigned if two edges are consecutively connected in the graph; otherwise a weak interaction $w$ is used to model crosstalk. (d) The yield is defined as the partition function of the target structure (the Hamiltonian path) divided by the total partition function of all possible assembled chains.}
    \label{fig:1}
\end{figure}

\begin{figure*}
    \centering
    \includegraphics[width = \linewidth]{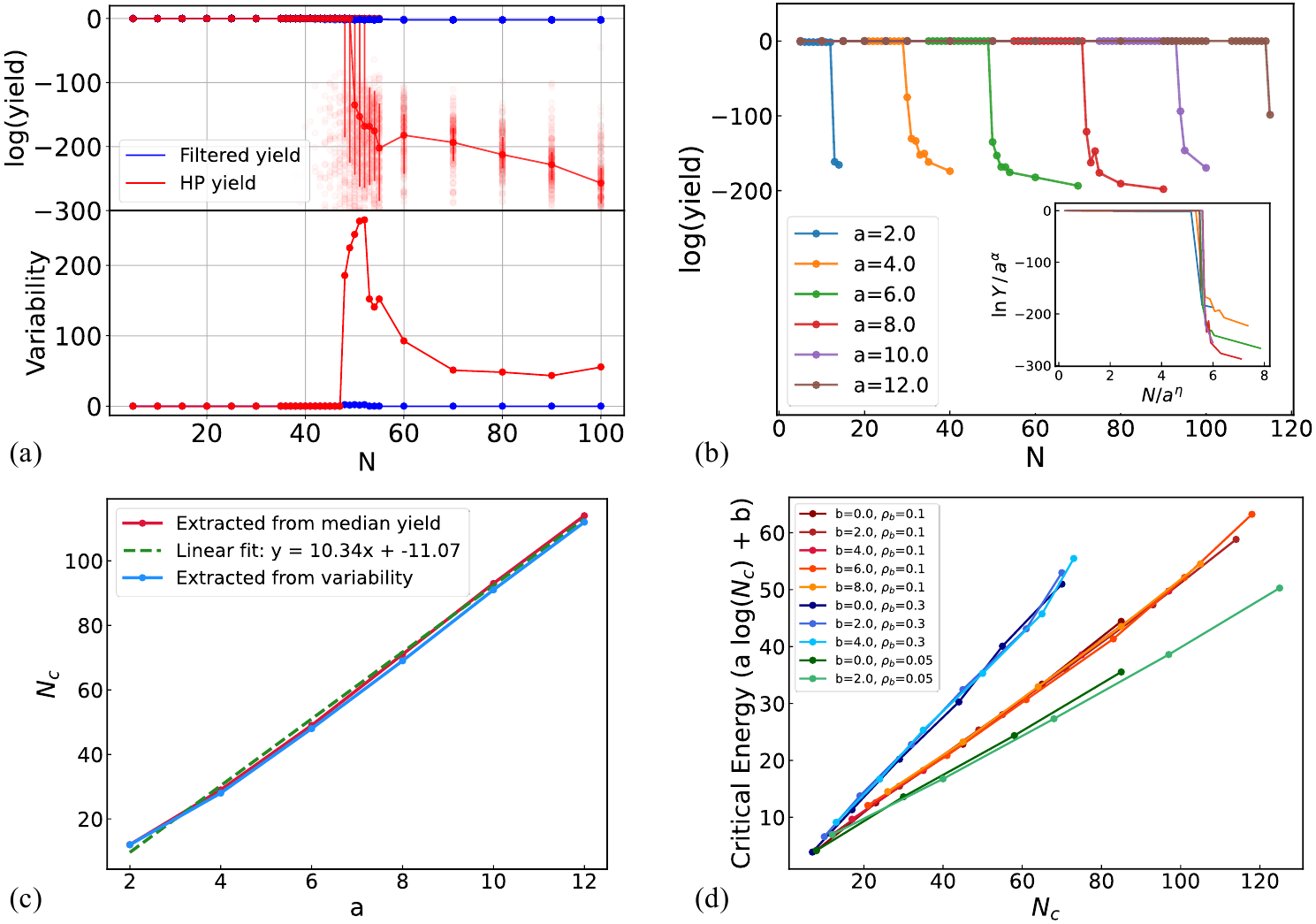}
    \caption{\justifying \textbf{Scaling of HP yield and phase transition-like behavior.} (a) Upper: Median log-yield of filtered assemblies (blue) and Hamiltonian path assemblies (red) as a function of graph size $N$ for $a=6.0$, $b=2.0$, $w = 0.0$, $\rho_b=0.1$. Points show individual graph instances; error bars indicate interquartile range (IQR). Lower: Variability of log-yields across graph instances for the same parameters. (b) HP yield curves as a function of graph size $N$ for various gap-scaling coefficients $a$. Other conditions are the same: $b=2.0$, $w = 0.0$, $\rho_b=0.1$. Inset: Curves for different $a$ collapse under the scaling form $\log Y_\text{HP}= a^\alpha F(N/a^\eta)$. (c) Critical graph size $N_c$ extracted from yield drop (red) and variability peaks (blue). Dashed line: best linear fit with $R^2=0.9983$. 
    (d) Critical energy gap $w-s=a\log N_c+b$ as a function of $N_c$. }
    \label{fig:2}
\end{figure*}

A key challenge in solving the HPP via the optimization framework in \cite{murugan2015undesired} is that the target structure -- the Hamiltonian path itself -- is unknown. This precludes direct optimization of its yield as in Fig.~\ref{fig:1}(d), as the correct path cannot be explicitly enumerated in advance.
Although the specific path is unknown, any Hamiltonian path must contain exactly $N-1$ edges in order to visit all $N$ nodes once and only once. Therefore, we can maximize a surrogate objective: the total yield of all linear assemblies of length $N-1$. We refer to this quantity as the filtered yield. 
This filtered yield includes the desired Hamiltonian path as well as competing structures of the same length.
Mathematically, the filtered yield is given by
\begin{equation}
    \label{eq:filtered_yield}
    Y_\text{filtered}= \frac{Z_{\text{length}=N-1}}{Z_\text{all}}=\frac{\sum_{i,j}c_i\left[\mathbf{B}^{N-2}\right]_{ij}}{\sum_{i,j}c_i\left[\left( \mathbb{I}- \mathbf{B}\right)^{-1}\right]_{ij}},
\end{equation}
where $\mathbf{B}$ is defined as in \eqref{eq:HP_yield}.
We then optimize the species concentrations $\{c_i\}$ to maximize the surrogate objective $\log(Y_{\text{filtered}})$, and use the resulting set of optimal $\{c_i^*\}$ to evaluate the true yield of the Hamiltonian path $Y_\text{HP}$ given in \eqref{eq:HP_yield}. 

We test and evaluate this method on randomly generated directed graphs. The graphs are generated as follows: (1) a single Hamiltonian path is constructed, in order to enable ground-truth yield evaluation; (2) additional extraneous edges are added with probability $\rho_b$; (3) these directed edges are only added to ``shortcut'' the Hamiltonian path, thus restricting the graphs to be acyclic (see an accompanying Python notebook \cite{notebook_git} for details). For general graphs with loops, there is an exponential proliferation of competing structures with all strong bonds; see additional results and discussion in \SI. Under these restrictions, the expected number of edges in an acyclic graph with $N$ nodes is given by 
\begin{equation}
N_{\text{edge}} = N-1 + \frac{\rho_b}{2}(N-1)(N-2).
\end{equation}
For each $N$ ranging from 5 to 100 nodes, we randomly generate 200 graphs following the above process and perform the optimization over concentrations for each graph independently. 

Fig.~\ref{fig:2}(a) shows the yield behavior as a function of graph size $N$ for a representative energy parameter choice $a=6.0$, $b=2.0$ (defined in \eqref{eq:e_params}). 
In the upper panel, we plot the median of the logarithmic filtered yield (blue) and the Hamiltonian path (HP) yield (red) over 200 graph instances. Individual points indicate yields for each realization, while the variability -- defined as the difference between the 75th and 25th quartiles -- is used to quantify fluctuations and is shown in the lower panel of Fig.~\ref{fig:2}(a). 
Beyond a critical threshold value of $N$, denoted by $N_c$, a sharp separation between filtered yield and HP yield emerges, with an abrupt drop in the HP yield. The variability also peaks sharply near the critical $N_c$, reflecting the coexistence of two nearly degenerate free-energy states -- ordered HP and disordered mis-bound configurations.

This critical transition becomes more apparent when examining the effect of varying the gap-scaling coefficient $a$, as shown in Fig.~\ref{fig:2}(b). To quantify the observed scaling, we collapse the yield curves using the ansatz:
\begin{equation}
    \label{eq:collapse}
    \log (Y_\text{HP})=a^\alpha  F\left(\frac{N}{a^\eta} \right)
\end{equation}
for some function $F$, with fitted exponents $\eta = 1.222$ ($R^2=1.000$) and $\alpha = -0.178$ ($R^2=0.485$), as shown in the inset of Fig.~\ref{fig:2}(c). The excellent collapse along the $N/a^\eta$ axis confirms that the transition is governed by a scaling law, resembling a first-order phase transition.

We extract the critical size $N_c$ for each $a$ based on a yield threshold criterion: $\log Y_\text{HP}<-10$. This threshold is robust, as the HP yield decreases rapidly and consistently near the transition. We can also independently extract $N_c$ as the point where the variability spikes. As shown in Fig.~\ref{fig:2}(c), the critical $N_c$ values extracted from the yield drop (red) agree closely with those inferred from the peak of the variability fluctuation curves (blue), providing independent validation. 
%
These results are consistent with a first-order–like transition, with the HP yield serving as the order parameter. The system exhibits competing free-energy minima: an energetically favored ordered HP configuration and an exponentially large manifold of competing structures favored entropically. As $N$ increases, the entropic manifold gains weight until, at $N=N_c$, the two basins become comparable.

To further understand this, in Fig.~\ref{fig:2}(d) we examine the critical energy gap $\Delta E_c = w-s=a\log N_c+b$ versus $N_c$ for varying $\rho_b$ and $b$ values. 
The critical energy gap required to reliably assemble the HP increases with both system size $N$ and $\rho_b$. The parameter $\rho_b$ controls the number of possible extraneous structures that can compete with the HP. All curves collapse by $\rho_b$, indicating that the density of extraneous edges sets the universal energy–entropy balance of this problem.
We derive the total partition function of all length-$N-1$ chains in closed form via a generating function approach (see details in \SI), and show the asymptotic scaling to be:
\begin{equation}
    a(N_c)\;=\;\frac{\Lambda(\rho_{\mathrm{eff}})}{\log N_c}\,N_c\;+\;O(1),
    \label{eq:anscaling}
\end{equation}
where $\Lambda(\rho)=2\log\!\big(1+\sqrt{\rho}\,\big)$, and \(\rho_{\mathrm{eff}}\) is an effective edge density that incorporates unequal particle concentrations through a graph-averaged normalization factor, and reduces to \(\rho_b\) when component concentrations are equal.
However, in the $N$-range we probed here the logarithmic correction is of order unity, so a simple linear fit effectively captures the data: the dashed line in Fig.~\ref{fig:2}(c) provides an excellent linear fit with $R^2 = 0.9983$. For simplicity, we use this linear approximation in the designability discussion that follows. 
On the other hand, Eq.~\eqref{eq:anscaling} also implies that the critical energy gap, $\Delta E_c$, scales linearly with $N_c$:
\begin{equation}
\Delta E_c \sim \Lambda(\rho_\text{eff}) N_c. 
\end{equation}
This explains the near linear dependence observed in Fig.~\ref{fig:2}(d), with slopes dependent on $\rho_b$.

\paragraph{Designability of interactions---}

\begin{figure*}[ht]
\centering
\begin{subfigure}{0.45\textwidth}
    \includegraphics[width=\textwidth]{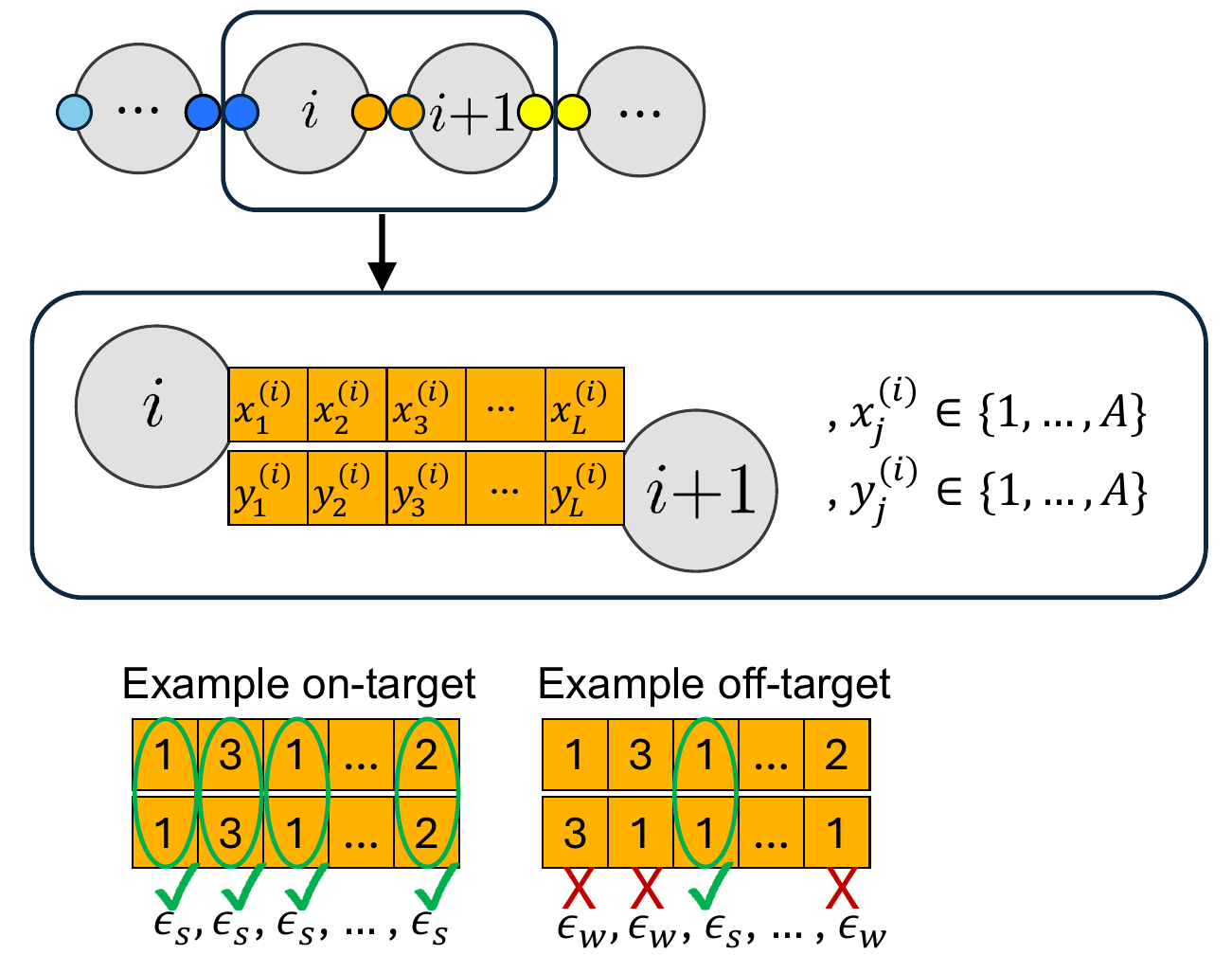}
    \vspace{0.0cm} 
    \caption{\label{fig:binding_setup}}
\end{subfigure}
\hfill
\begin{subfigure}{0.54\textwidth}
    \includegraphics[width=\textwidth]{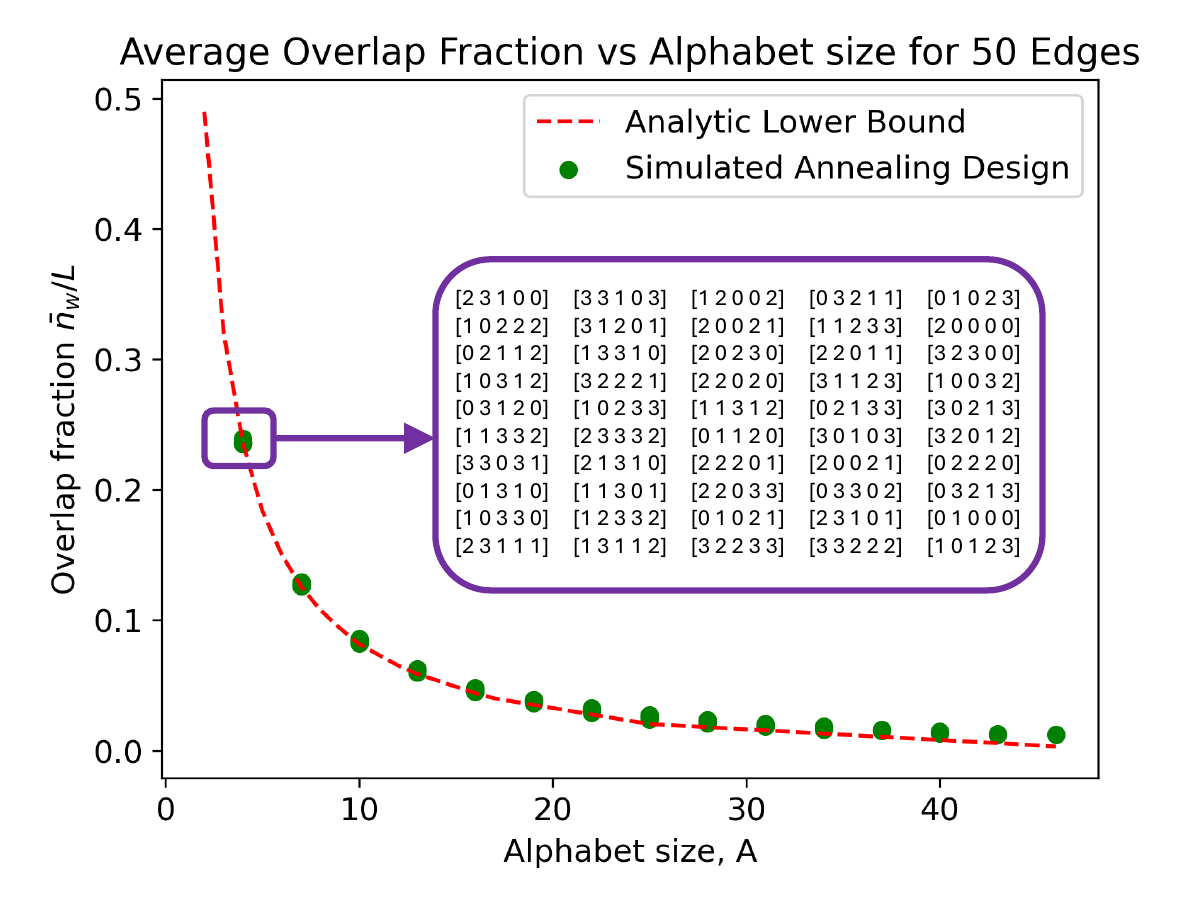}
        \caption{\label{fig:binding_ana_vs_sim}
}
\end{subfigure}
\caption{\label{fig:fig3} \justifying \textbf{Binding design scheme and agreement with analytic expression.} (a) Binding sites between particles are modeled via a lock-key pair of chains of length $L$ with components chosen from an alphabet of size $A$. On-target ``letters'' have an average energy of $\epsilon_s$ and off-target components an average energy of $\epsilon_w$, to model crosstalk. Examples of an on-target binding pair and an off-target pair are shown, where $A=3$ and $L=5$. The on-target binding
has total strength $s = \epsilon_s L$. (b) Overlap fraction $\frac{\bar{n}_w}{L}$ as a function of alphabet size, $A$. The analytic lower bound \eqref{eq:overlap fraction} and a numerical simulated annealing protocol show good agreement. An example of bindings for $A=4$ and $L=5$ generated via simulated annealing is shown.}
\end{figure*}

Now that we have a protocol for optimization without prior knowledge of the solution, we must ask if the interactions required to encode the HPP are designable, at least in principle. In order to encode and solve the HPP in a molecular system, one must be able to design $N_{\text{edge}}$ distinct bindings for each of the graph edges while ensuring that the difference between on-target and off-target bindings, $w-s$, is sufficiently large.

We model potential bindings as a lock-key pair of chains, each of length $L$ with components drawn from an alphabet of size $A$ containing the ``letters'' $\{1, \ldots, A\}$. A diagram of this setup with locks denoted by $x^{(i)}$ and keys denoted by $y^{(i)}$ is shown in Fig.~\ref{fig:fig3}(a). Suppose that on-target bonds between matching letters have average energy $\epsilon_s<0$, and off-target bonds have binding energy $\epsilon_w<0$, with $|\epsilon_w|<|\epsilon_s|$. Then, the average strong binding for a chain of length $L$ is
$s=\epsilon_s L$ and the average weak binding energy is $w = \epsilon_s \bar{n}_w + \epsilon_w(L-\bar{n}_w)$, where $\bar{n}_w$ is the average number of matching ``letters'' between two off-target chains. 

Thus, in order to encode the HPP into a system with alphabet size $A$, one must find $N_{\text{edge}}$ bindings of some length $L$ so that the energy gap between on- and off-target binding satisfies
\begin{equation} \label{eq:wsbound_L}
w - s = \epsilon_s \bar{n}_w + \epsilon_w (L-\bar{n}_w) - \epsilon_s L \geq a \log(N) + b.
\end{equation}

To consider self-assembly processes over experimentally realistic timescales, we require that $s$ is fixed below some threshold. With a fixed $|s|$, we must then require $\epsilon_s = s/L$, so that $\epsilon_s$ is smaller for larger $L$. A similar scheme is presented for magnetic systems in \cite{du2022programming}, with $s = -10 k_B T$. With fixed strong bonding $s$, the ``encodability" bound \eqref{eq:wsbound_L} becomes 
\begin{equation} \label{eq:fixed_s_bound}
|s| \left(1- \frac{\bar{n}_w}{L}\right) \left(1 - \frac{\epsilon_w}{\epsilon_s}\right) \geq a\log(N) + b.
\end{equation}
For a molecular system with the three design parameters of energy fraction $\frac{\epsilon_w}{\epsilon_s}$, overlap fraction $\frac{\bar{n}_w}{L}$, and the maximum allowed $s$, \eqref{eq:fixed_s_bound} yields the maximum allowed $a$ and $b$ values that can be realized for a given $N$. 

A lower bound for the parameter $\frac{\bar{n}_w}{L}$ can be found in terms of the alphabet size $A$ by (see SI for details):
\begin{equation}\label{eq:overlap fraction}
    \frac{\bar{n}_{w}}{L} = \frac{(A-r)q^2 + r(q+1)^2 - N_{\text{edge}}}{N_{\text{edge}}(N_{\text{edge}}-1)} ,
\end{equation}
where $q = \lfloor N_{\text{edge}}/A \rfloor$ and $r = N_{\text{edge}} \bmod A$. When $N_{\text{edge}}$ is divisible by $A$ this simplifies to $\frac{\bar{n}_{w}}{L} = \frac{1}{N_{\text{edge}}-1}\left(\frac{N_{\text{edge}}}{A} - 1\right)$. Fig.~\ref{fig:fig3}(b) shows good agreement between the analytic lower bound \eqref{eq:overlap fraction} and a numerical simulated annealing protocol. In the limit where $N_{\text{edge}}\gg A$, this lower bound becomes $\frac{\bar{n}_{w}}{L} \sim \frac{1}{A}$. This makes sense as given a particular ``lock'' component $x^{(i)}_k$, there is a probability of $\frac{1}{A}$ that an off-target ``key'' component, $y^{(j)}_k$, will match. Then, the total expected number of overlaps between off-target pairs $x^{(i)}$ and $y^{(j)}$ is $\bar{n}_w = \frac{L}{A}$.\\

\paragraph{Phase Diagram for Encoding and Solving HPP Problems---}

\begin{figure*}
    \centering
    \includegraphics[width = 0.8\textwidth]{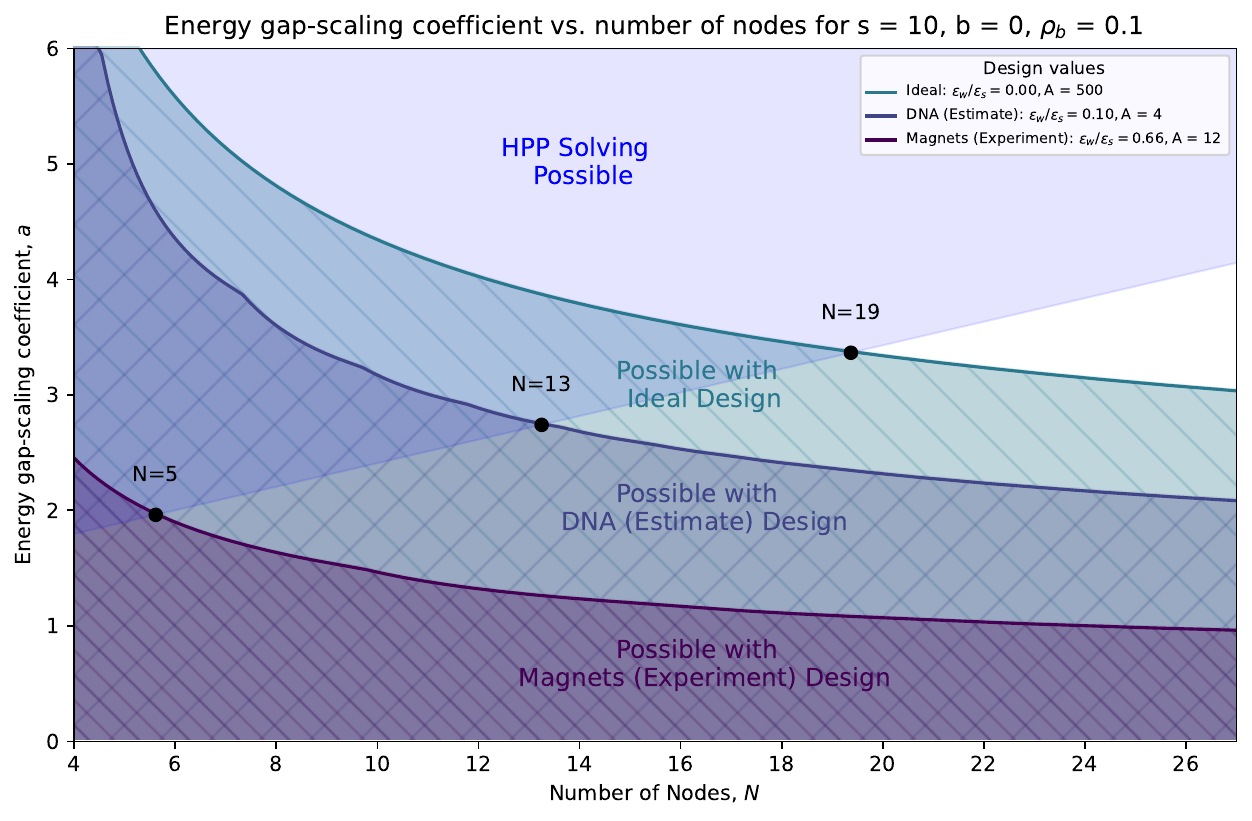}
    \caption{\justifying \textbf{Phase diagram for designable molecular systems that can solve HPP problems.} Phase diagram for encoding and solving an HPP in equilibrium, with parameters $b=0$, $s=10$, $\rho_b=0.1$. For the ideal case with large $A$ and $\frac{\epsilon_w}{\epsilon_s} = 0$, the largest solvable graph has 19 nodes. 
    Using a heuristic for DNA base pairs as the letters, with $A=4$ and $\epsilon_w / \epsilon_s = 0.1$, the largest graph has 13 nodes. Using the magnetic dipole system from \cite{du2022programming} as the ``letters'', with $A=12$ and $\frac{\epsilon_w}{\epsilon_s} = 0.66$, the largest solvable graph has 5 nodes.}
    \label{fig:phase_diagram}
\end{figure*}

Using the results of the previous two sections, we can now determine the largest graph size for an HPP to be encoded and solved by a given molecular system in equilibrium. 
By imposing equality in \eqref{eq:fixed_s_bound} and substituting the lower bound on $\frac{\bar{n}_w}{L}$ from \eqref{eq:overlap fraction}, we obtain the maximal attainable values of $a$ and $b$ for any molecular system with alphabet size $A$ and off-target energy fraction $\frac{\epsilon_w}{\epsilon_s}$. In Fig.~\ref{fig:phase_diagram}, we provide a phase diagram showing regions where the encoding of an HPP is possible and regions where the HPP is solvable. The intersection of these regions is where it is feasible that a given designed molecular system can encode and solve an HPP. 

The phase diagram in Fig.~\ref{fig:phase_diagram} shows regions for an ideal molecular design, a synthetic magnetic panel system from \cite{du2022programming}, and a toy model of DNA. The ideal design features a very large alphabet size $A$ and zero crosstalk between letters, $\frac{\epsilon_w}{\epsilon_s} = 0$, maximizing the attainable energy gap in \eqref{eq:fixed_s_bound}. Even in this ideal case, however, the molecular system can solve the HPP only up to a finite graph size. For the representative parameters $s=10$, $b=0$, and $\rho_b=0.1$, the maximum solvable graph size is $N=19$.

In the magnetic system \cite{du2022programming}, sets of magnetic dipole panels were designed to minimize crosstalk between off-target panels. The phase diagram yields an estimate of the largest HPP that could be solved with their 12-panel system, treating the magnetic panels as the system’s ``letters'' ($A=12$). For the experimentally realized design with 3 dipoles per panel -- corresponding to an off-target energy fraction of $\frac{\epsilon_w}{\epsilon_s} = 0.66$ -- the maximum designable and solvable graph problem has 5 nodes. 

For our estimated DNA model, we use an alphabet size of $A=4$, corresponding to the canonical base pairs. To estimate the off-target energy fraction $\epsilon_w/\epsilon_s$, we use the worst case overlap between off-target base pairs from \cite{peyret_nearest-neighbor_1999}, namely Guanine-Guanine with binding energy $-2.22$ kcal/mol. 
On-target base pair energies are $-12.1$ kcal/mol for A–T and $-21.0$ kcal/mol for G–C \cite{Yanson1979}.
Taking the latter as the representative on-target value, we obtain an estimate of $\epsilon_w/\epsilon_s \approx 0.1$ for a DNA-like system. We note that our DNA-like system is a toy model: it does not instantiate the actual base pair binding energies but instead fixes the total strong bonding energy $s$ while using the estimated ratio $\epsilon_w/\epsilon_s$. 

Table \ref{tab:HPPexamples} lists further example systems with the largest possible solvable HPP problem for each case. Phase diagrams and additional results for these systems (along with the ``Pacman'' particles from \cite{sacanna_lock_2010}) can be found in \SI. 

In both the ideal case and in these specific examples, Fig.~\ref{fig:phase_diagram} demonstrates that for this problem, there is a fundamental tension between designability and solvability. With increasing $N$, a larger number of components need to be designed, so the maximum energy gap between on- and off- target components (governed by $a$) decays.  On the other hand, for the HPP to be encoded and solved, $a$ must increase with $N$ (as shown in Fig.~\ref{fig:2}(c)). These competing constraints ensure that even in the ideal design case, there is a stopping point for the largest possible solvable problem in equilibrium.  Moreover, this is a best case scenario bound: our model for component design assumes ``lock'' and ``key'' binding strings exactly overlap, and the method we outline for solving HPP problems by optimizing for the filtered yield holds only for acyclic graphs. The existence of a maximum possible problem size, even in this best case scenario, motivates that molecular computing in equilibrium cannot scale robustly enough to solve larger and more complex problems. We propose that to do this, one will need to include out-of-equilibrium protocols, such as imposing constraints on particle supply \cite{kurtz1996active} or incorporating kinetic proofreading mechanisms \cite{hopfield1974kinetic, zhu2024proofreading, liang2025magnetic}.

\setlength{\tabcolsep}{8pt}   
\renewcommand{\arraystretch}{1.20}

\begin{table*}[t]
  \centering
  \caption{\label{tab:HPPexamples} Design and HPP-solving parameters for several example systems and the maximum size $N$ that can be attained}
  \begin{tabular}{@{} l c c c c c c @{}}  
    \toprule
    \multicolumn{1}{c}{\textbf{System}} &
    \multicolumn{3}{c}{\textbf{Design}} &
    \multicolumn{3}{c}{\textbf{HPP}}\\
    \cmidrule(r){2-4}\cmidrule(l){5-7}
    & $A$ &
      Mean $\epsilon_w / \epsilon_s$ &
      $ \lvert s \rvert \ [k_B T]$ &
      $b$ &
      $\rho_b$ &
      Max.\ $N$ \\
    \midrule
    Ideal System                       & $\infty$ & 0   & 10 & 0 & 0.05 & 22   \\
    Ideal System                       & $\infty$ & 0   & 10 & 0 & 0.1 & 19   \\
    Ideal System                       & $\infty$ & 0   & 10 & 0 & 0.3 & 15   \\
    
    3-Dipole Magnet Panels  ($b=0$)     & 12       & 0.66& 10 & 0 & 0.1 & 5    \\
    3-Dipole Magnet Panels  ($b=2$)     & 12       & 0.66& 10 & 2 & 0.1 & None    \\
    
    25-Dipole Magnet Panels ($b=0$)    & 12       & 0.33& 10 & 0 & 0.1 & 13   \\
    25-Dipole Magnet Panels ($b=2$)    & 12       & 0.33& 10 & 2 & 0.1 & 11    \\
    
    DNA (estimate, $\rho_b=0.05$)       & 4     & 0.1& 10 & 0 & 0.05 & 14   \\
    DNA (estimate, $\rho_b=0.1$)       & 4     & 0.1& 10 & 0 & 0.1 & 13   \\
    
    ``Pac-Man’’ Colloidal Particles (estimate, $\rho_b=0.1$)    & 6        & 0.2& 10 & 0 & 0.1 & 13    \\
    ``Pac-Man’’ Colloidal Particles (estimate, $\rho_b=0.3$)    & 6        & 0.2& 10 & 0 & 0.3 & 10    \\
    \bottomrule
  \end{tabular}
\end{table*}

\paragraph{Acknowledgments---} 
The authors thank Francesco Mottes for helpful discussions and draft comments. This work was supported by the NSF AI Institute of Dynamic Systems (2112085) and the Alfred P. Sloan Foundation under grant No. G-2021-14198 and an NSERC PGS-D fellowship to EC.

\bibliography{bib}

\end{document}


\title{Supplemental Material for:\\ Fundamental Scaling Constraints for Equilibrium Molecular Computing}

\author{Erin Crawley$^{\S,}$}
\affiliation{School of Engineering and Applied Sciences, Harvard University, Cambridge, Massachusetts 02138, USA} 
\affiliation{Department of Physics, Harvard University, Cambridge, Massachusetts 02138, USA}

\author{Qian-Ze Zhu$^{\S,}$}
\affiliation{School of Engineering and Applied Sciences, Harvard University, Cambridge, Massachusetts 02138, USA} 

\author{Michael P. Brenner}
\affiliation{School of Engineering and Applied Sciences, Harvard University, Cambridge, Massachusetts 02138, USA} 
\affiliation{Department of Physics, Harvard University, Cambridge, Massachusetts 02138, USA}

\date{\today}
\begin{abstract}
\begin{description}
\item[\S] These authors contributed equally to this work.
\end{description}
\end{abstract}
\maketitle

\section{Optimization for filtered yield in graphs with loops}

In this section, we extend the optimization framework introduced in the main text to general graphs that include loops. The goal is to evaluate whether maximizing the filtered yield remains a viable strategy when cycles are permitted in the graph structure.

Here, we generate graphs following principles similar to those described in the main text, but without the ``no-loops" restrictions: (1) a Hamiltonian path is embedded by construction to enable ground-truth yield evaluation; (2) additional extraneous edges are included independently with probability $\rho_b$. 
Unlike in the main text, loops (i.e., directed cycles) are allowed, leading to cyclic graphs. Under these conditions, the expected number of edges in a graph with $N$ nodes is given by 
\begin{equation}
N_{\text{edge}} = N-1+\rho_b N(N-2).
\end{equation}
which accounts for the edges of the Hamiltonian path plus randomly added edges between all other node pairs (excluding self-loops from one node back to itself).

We adopt the same optimization objective as in the acyclic case, maximizing the filtered yield
\begin{equation}
    \label{eq:filtered_yield}
    Y_\text{filtered}= \frac{Z_{\text{length}=N-1}}{Z_\text{all}}=\frac{\sum_{i,j}c_i\left[\mathbf{B}^{N-2}\right]_{ij}}{\sum_{i,j}c_i\left[\left( \mathbb{I}- \mathbf{B}\right)^{-1}\right]_{ij}}
\end{equation}
where $B_{ij} = e^{-\beta E_{ij}} c_j$. As in the main text, after optimizing the species concentration ${c_i}$ to maximize $\log(Y_{\text{filtered}})$, we use the set of optimal ${c_i^*}$ to evaluate the true yield of the Hamiltonian path $Y_\text{HP}$:
\begin{equation}
    \label{eq:HP_yield}
    Y_\text{HP} = \frac{Z_\text{HP}}{Z_\text{all}}=\frac{\prod_{i \in \text{HP}} c_i e^{-(N-2) \beta s}}{\sum_{i,j}c_i\left[\left( \mathbb{I}- \mathbf{B}\right)^{-1}\right]_{ij}}
\end{equation}

For each graph size $N$, we randomly generate 20 (potentially cyclic) graphs  according to the criteria above and perform the optimization independently. For larger graph sizes $N>30$, we frequently observe numerical instabilities that hinder convergence. Additionally, in rare cases, yield calculations become unreliable and can exceed unity, indicating numerical overflow, though this does not qualitatively affect our conclusions.

Fig.~\ref{fig:SI1} presents the results of this analysis for various gap coefficients $a$, with fixed $b=2.0$, $w=0.0$ and $\rho_b=0.1$. 
The median log-yields of both the filtered yield (blue) and HP yield (red) are plotted across graph sizes, along with individual instance results. Fluctuations are again quantified by the interquartile range.
\begin{figure*}
    \centering
    \includegraphics[width = \linewidth]{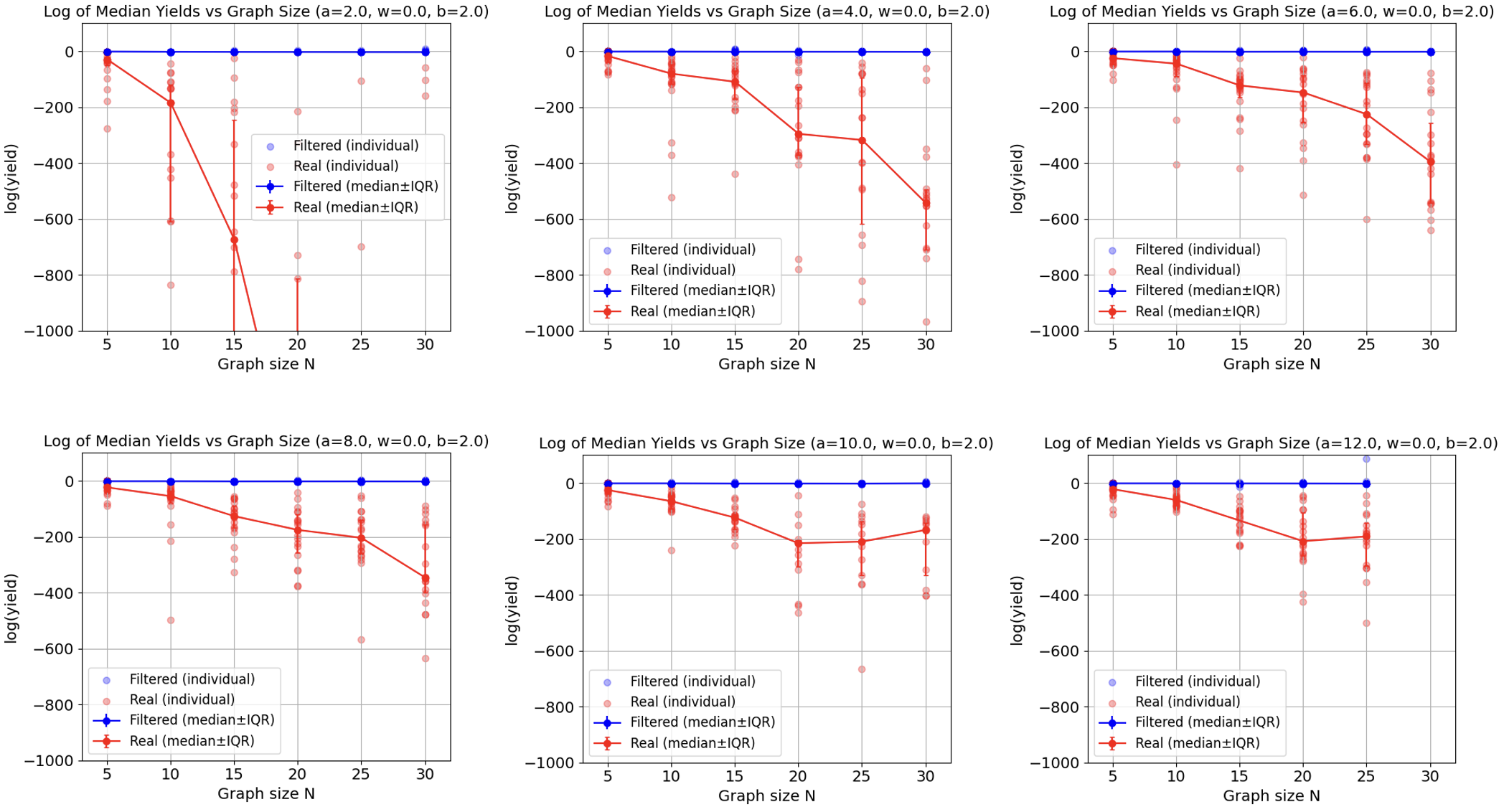}
    \caption{\justifying \textbf{Filtered yield optimization in graphs with loops.} Median log-yield of filtered assemblies (blue) and Hamiltonian path assemblies (red) as a function of graph size $N$, shown for varying gap coefficients $a$ and fixed $b=2.0$, $w = 0.0$, $\rho_b=0.1$. Points show individual graph instances; error bars indicate interquartile range (IQR).}
    \label{fig:SI1}
\end{figure*}
Unlike in the acyclic case, we observe that the HP yield remains vanishingly small -- well below $e^{-10}$ -- for all values of $N$ and $a$ tested. This indicates that the optimization strategy fails to recover the Hamiltonian path in graphs with loops, even when its length is explicitly encoded in the objective. The failure stems from a fundamental limitation of the grand canonical ensemble in the presence of cycles. When loops are allowed and each particle species is present in unlimited supply, the system can explore an exponentially large number of configurations composed entirely of strong bonds. Loops enable the reuse of species to form arbitrarily long, strongly connected chains. As a result, the denominator $Z_{\text{all}}$ becomes dominated by such cyclic, misassembled structures that are indistinguishable energetically from the target, while the Hamiltonian path -- being just one of exponentially many possibilities -- has negligible statistical weight. 

These results conclude that the optimization protocol presented in the main text is only valid for acyclic graphs. For graphs with loops, alternative strategies -- such as imposing constraints on particle supply or incorporating kinetic proofreading mechanisms -- may be necessary to recover high-fidelity self-assembly of Hamiltonian paths.

\section{Analytical analysis of the critical energy–gap scaling}

\subsection*{Physical background and setup}
At the transition reported in the main text, the system balances \emph{energy} (favoring the target Hamiltonian Path, HP) against \emph{entropy} (the exponential number of competing incorrect assemblies). At the critical energy, $\Delta E_c$, the partition function of the HP becomes comparable to that of the dominant competing structures. 
Here we present an analytic argument for how the critical energy scales with the graph size $N$. For clarity, we consider the case where all components have equal concentrations; this is the regime in which a closed–form derivation is tractable. Unequal concentrations only renormalize prefactors and do not affect the asymptotic scaling with $N$.

As noted in the main text, random acyclic graphs are generated as follows: (i) fix a single Hamiltonian path (HP) to enable ground-truth yield evaluation; (ii) add each additional directed edge independently with probability \(\rho_b\); (iii) allow only “shortcut’’ edges relative to the HP so that the graph remains acyclic. We can encode edge probabilities for acyclic graphs in a strictly upper–triangular matrix $A\in\mathbb{R}^{N\times N}$,
\begin{equation}
\label{eq:Amatrix}
A=\begin{bmatrix}
0 & 1 & \rho & \rho & \cdots & \rho \\
0 & 0 & 1 & \rho & \cdots & \rho \\
0 & 0 & 0 & 1 & \cdots & \rho \\
0 & 0 & 0 & 0 & \ddots & \vdots \\
\vdots & \vdots & \vdots & \vdots & \ddots & 1 \\
0 & 0 & 0 & 0 & \cdots & 0
\end{bmatrix},
\end{equation}
where, for notational simplicity, we set \(\rho_b\equiv\rho\in(0,1)\) and \(a_{ij}\) denotes the probability that the directed edge \(i\to j\) exists. The probability weight of a length-\(k\) chain \((i_1\!\to\! j_1),(i_2\!\to\! j_2),\ldots,(i_k\!\to\! j_k)\) is then \(a_{i_1j_1}a_{i_2j_2}\cdots a_{i_kj_k}\). The ordered HP has a weight of 1.

To incorporate the energy difference between strong and weak bonds, we assign a penalty \(-\Delta E=-(w-s)=-a\log N-b\) to each weak bond (absorbing \(\beta\) into \(\Delta E\)). We encode this with a “soft Kronecker’’ matrix \(B\in\mathbb{R}^{N\times N}\),
\begin{equation}
B=\begin{bmatrix}
1 & \gamma & \gamma & \gamma & \cdots & \gamma \\
\gamma & 1 & \gamma & \gamma & \cdots & \gamma \\
\vdots & \vdots & \vdots & \vdots & \ddots & \gamma \\
\gamma & \gamma & \gamma & \gamma & \cdots & 1
\end{bmatrix},
\qquad \gamma:=e^{-\Delta E},
\end{equation}
so that if the end of one edge is \(j\) and the start of the next edge is \(i'\), the energetic factor is \(b_{j i'}=1\) when \(j=i'\) and \(b_{j i'}=\gamma\) otherwise. Equivalently, letting $\mathbf{1}$ denote a length $N$ vector with entries all equal to 1, \(B(\gamma)=\gamma J+(1-\gamma)I\) with \(J=\mathbf{1}\mathbf{1}^\top\).

The resulting \emph{relative} partition function of all assembled chains of length \(N-1\) is
\begin{equation}
    \boxed{\;
    S_{N-1}=\sum_{i_1,j_1,\ldots,i_{N-1},j_{N-1}}
    a_{i_1j_1}\,b_{j_1i_2}\,a_{i_2j_2}\cdots
    b_{j_{N-2}i_{N-1}}\,a_{i_{N-1}j_{N-1}}
    \;}
    \label{eq:sum}
\end{equation}
while the ordered HP contributes \(1\) in this normalization.
Our goal is to evaluate $S_{N-1}$ exactly for finite $N$, use the condition $S_{N-1}\sim\mathcal{O}(1)$ to obtain the an expression for the critical energy gap $\Delta E_c$, and thereby determine the large-$N$ scaling of $\Delta E_c$. 

\subsection*{Matrix reduction and generating function}
Summing over internal indices in \eqref{eq:sum} yields the exact matrix form:
\begin{equation}
    \boxed{\;
    S_{N-1}=\mathbf{1}^\top\big(A\,B(\gamma)\big)^{N-2}A\,\mathbf{1}
    =\mathbf{1}^\top\big((1-\gamma)A+\gamma(A\mathbf{1})\mathbf{1}^\top\big)^{N-2}A\,\mathbf{1}.
    \;}
    \label{eq:openchain-matrix}
\end{equation}
Define $M:=(1-\gamma)A$, $u:=\gamma A\mathbf{1}$, and $v:=\mathbf{1}$. We introduce the generating function:
\begin{equation}
    S(z):=\sum_{k\ge1}S_k\,z^k
    =\mathbf{1}^\top\big(I-z(M+uv^\top)\big)^{-1}zA\,\mathbf{1}.
\end{equation}
The Sherman–Morrison formula gives
\[
(I-z(M{+}uv^\top))^{-1}=R+\frac{z\,Ruv^\top R}{1-z\,v^\top Ru},\qquad R:=(I-zM)^{-1}, 
\]
we thus obtain
\[
S(z)=z\,\mathbf{1}^\top RA\mathbf{1}
+\frac{z^2(\mathbf{1}^\top Ru)(v^\top RA\mathbf{1})}{1-z\,v^\top Ru}.
\]
Since $u=\gamma A \mathbf{1}$, $v=\mathbf{1}$, $M=(1-\gamma)A$, the scalars in the above expression can be evaluated as:
\[
\mathbf{1}^\top RA\mathbf{1}=\Theta(z),
\qquad
v^\top Ru=\gamma\,\Theta(z),
\]
where we have defined
\begin{equation}
    \boxed{\;
    \Theta(z):=\mathbf{1}^\top\big(I-z(1-\gamma)A\big)^{-1}A\,\mathbf{1}
    =\sum_{m\ge1}(1-\gamma)^{m-1}g_m\,z^{m-1},
    \quad g_m:=\mathbf{1}^\top A^m\mathbf{1}.
    \;}
    \label{eq:Theta}
\end{equation}
Therefore
\begin{equation}
    \boxed{\;
    S(z)=\frac{z\,\Theta(z)}{1-z\,\gamma\,\Theta(z)}.
    \;}
    \label{eq:S-GF}
\end{equation}

\subsection*{From the generating function to a composition formula}
We would like to obtain an expression for $S_{N-1}$. To do this, we will extract coefficients $S_k$ of $z^k$ from \eqref{eq:S-GF} using standard formal power series rules. 
For clarity, we use the coefficient extractor notation $[z^k]F(z)$ to denote the coefficient of $z^k$ in the (formal) power series $F(z)$. We also recall the Cauchy product identity
\begin{equation}\label{eq:Cauchy}
[z^m]\Big(\sum_{i\ge0} a_i z^i\Big)\Big(\sum_{j\ge0} b_j z^j\Big)=\sum_{i=0}^{m} a_i b_{m-i}.
\end{equation}
First, we expand the denominator of \eqref{eq:S-GF} as a geometric series:
\[
\frac{1}{1-z\gamma\Theta(z)}=\sum_{q\ge0}\big(z\gamma\Theta(z)\big)^q,
\qquad
S(z)=\sum_{p\ge1}\gamma^{p-1}z^{p}\,\Theta(z)^{p}.
\]
Extracting the coefficient of $z^k$ yields:
\begin{align} \label{eq_Sk_expression_Theta}
S_k
&=[z^k]\,S(z)
=\sum_{p=1}^{k}\gamma^{p-1}\,\left([z^{k}]\,\Big(z^{p}\Theta(z)^{p}\Big)\right)
=\sum_{p=1}^{k}\gamma^{p-1}\,\left([z^{k-p}]\,\Big(\Theta(z)^{p}\Big)\right).
\end{align}
Now, we expand $\Theta(z)^p$ using the Cauchy product \eqref{eq:Cauchy} repeatedly. Given $\Theta(z)$ from \eqref{eq:Theta}, we have
\[
\Theta(z)^p =\sum_{\substack{m_1+\cdots+m_p=k\\ m_i\ge1}}\Big(\prod_{i=1}^{p}g_{m_i}\Big)\, (1-\gamma)^{k-p}z^{k-p}.
\]
Therefore,
\[
[z^{k-p}]\,\Theta(z)^p
=\,(1-\gamma)^{k-p}\!\! \sum_{\substack{m_1+\cdots+m_p=k\\ m_1, \ldots, m_p \ge1}}
\ \prod_{i=1}^{p} g_{m_i}.
\]
Plugging this into \eqref{eq_Sk_expression_Theta} gives the convolutional form for $S_k$:
\begin{equation}
    \boxed{\;
    S_k=\sum_{p=1}^{k}\gamma^{p-1}(1-\gamma)^{k-p}
    \sum_{\substack{m_1+\cdots+m_p=k\\ m_i\ge1}}
    \ \prod_{i=1}^p g_{m_i}.
    \;}
    \label{eq:comp-form}
\end{equation}

\subsection*{Powers of $A$ and a closed form for $g_m=\mathbf{1}^\top A^m\mathbf{1}$}
Let $\mathbf{N}$ denote the upper shift matrix of size $N\times N$ with $(\mathbf{N})_{i,i+1}=1$ and zeros elsewhere:
\begin{equation}
N=\begin{bmatrix}
0 & 1 & 0 & 0 & \cdots & 0 \\
0 & 0 & 1 & 0 & \cdots & 0 \\
0 & 0 & 0 & 1 & \cdots & 0 \\
0 & 0 & 0 & 0 & \ddots & \vdots \\
\vdots & \vdots & \vdots & \vdots & \ddots & 1 \\
0 & 0 & 0 & 0 & \cdots & 0
\end{bmatrix},
\end{equation}
Then, defining 
\begin{equation}
    P(x):=x+\rho\left(x^2+\cdots+x^{N-1}\right)=x+\rho\,x^2\frac{1-x^{\,N-2}}{1-x},
\end{equation}
we have
\begin{equation}
    A=P(\mathbf{N})=\mathbf{N}+\rho\sum_{q=2}^{N-1}\mathbf{N}^q,
    \qquad
    (\mathbf{N})^N=0.
\end{equation}
As a polynomial in $\mathbf{N}$, every power of $A$ is an upper triangular matrix with constant diagonals (an upper-Toeplitz matrix). That is, powers of $A$ take the form:
\begin{equation}
    A^m=P(\mathbf{N})^m=\sum_{d=m}^{N-1} c_d^{(m)}\,\mathbf{N}^{d}, 
\end{equation}
where the sum truncates at $d=N-1$, since $\mathbf{N}$ is nilpotent with $\mathbf{N}^N =0$. Since $A^m$ is an upper-Toeplitz matrix, we have
\begin{equation}
    (A^m)_{ij}= \begin{cases}
    c^{(m)}_{j-i} \, , \, j-i\ge m\\ 
    0\, , \, \text{otherwise}.
    \end{cases}
\end{equation}
To derive the close form of the coefficients $c_d^{(m)}$, we have for $P(\mathbf{N})$:
\[
P(\mathbf{N})=\mathbf{N}\Big(1+\rho\,(\mathbf{N}+\mathbf{N}^2+\cdots)\Big).
\]
up to degree $<N$. We may replace the finite tail by the infinite geometric tail since the nilpotency of $\mathbf{N}$ truncates the infinite sum. Then, 
\[
P(\mathbf{N})^m
=\mathbf{N}^m\sum_{s=0}^{m}\binom{m}{s}\,\rho^{\,s}\,\big(\mathbf{N}+\mathbf{N}^2+\cdots\big)^{s}
=\mathbf{N}^m\sum_{s=0}^{m}\binom{m}{s}\,\rho^{\,s}\,\mathbf{N}^{s}(1-\mathbf{N})^{-s}.
\]
Thus the coefficient of $\mathbf{N}^{d}$ in $P(\mathbf{N})^m$ equals
\[
c_d^{(m)}=\sum_{s=0}^{m}\binom{m}{s}\rho^{s}\,\left([\mathbf{N}^{\,d-m-s}](1-\mathbf{N})^{-s}\right).
\]
Using the binomial series $(1-\mathbf{N})^{-s}=\sum_{q\ge0}\binom{s+q-1}{s-1}\mathbf{N}^{q}$, we obtain
\[
[\mathbf{N}^{\,d-m-s}](1-\mathbf{N})^{-s}=\binom{d-m-1}{s-1},
\]
which yields
\begin{equation}
    \boxed{\;
    c_d^{(m)}=\sum_{s=0}^{\min(m,\,d-m)}
    \binom{m}{s}\,\rho^{\,s}\,\binom{d-m-1}{s-1},
    \quad d\ge m,
    \;}
    \label{eq:cdm}
\end{equation}
with the convention that the $s=0$ term contributes $1$ only when $d=m$.
Furthermore, since $\mathbf{1}^\top \mathbf{N}^d\mathbf{1}=N-d$ (there are $N\!-\!d$ ones on the $d$-th superdiagonal), we obtain
\begin{equation}
    \boxed{\;
    g_m=\mathbf{1}^\top A^m\mathbf{1}
    =\sum_{d=m}^{N-1}(N-d)\,c_d^{(m)}.
    \;}
    \label{eq:gm-sum}
\end{equation}
We may evaluate \eqref{eq:gm-sum} by inserting \eqref{eq:cdm}, separating the $s=0$ term, and interchanging the order of summation for $s\ge1$:
\begin{align*}
g_m
&=(N-m)\cdot 1
\;+\;\sum_{d=m}^{N-1}(N-d)\sum_{s=1}^{\min(m,\,d-m)}\binom{m}{s}\rho^s\binom{d-m-1}{s-1}\\
&=(N-m)\;+\;\sum_{s=1}^{m}\binom{m}{s}\rho^s
\sum_{d=m+s}^{N-1}(N-d)\binom{d-m-1}{s-1}.
\end{align*}
In the inner sum, set $t=d-m$ and $T=N-1-m$. Then $t$ runs from $s$ to $T$, and $N-d=N-(m+t)=(N-m)-t$. Thus
\[
\sum_{d=m+s}^{N-1}(N-d)\binom{d-m-1}{s-1}
=\sum_{t=s}^{T}\Big[(N-m)-t\Big]\binom{t-1}{s-1}
=(N-m)\!\sum_{t=s}^{T}\binom{t-1}{s-1}-\sum_{t=s}^{T}t\binom{t-1}{s-1}.
\]
We now apply two standard identities.

\smallskip
\noindent\textit{Hockey-stick identity 1.}
\[
\sum_{t=s}^{T}\binom{t-1}{s-1}=\binom{T}{s}.
\]

\noindent
Using $\displaystyle \binom{t}{s}=\frac{t}{s}\binom{t-1}{s-1}$, we have
\[
t\binom{t-1}{s-1}=s\binom{t}{s}.
\]
\noindent\textit{Hockey-stick identity 2.}
\[
\sum_{t=s}^{T}\binom{t}{s}=\binom{T+1}{s+1}.
\]

\smallskip
Applying these identities yields
\[
\sum_{t=s}^{T}\Big[(N-m)-t\Big]\binom{t-1}{s-1}
=(N-m)\binom{T}{s}-s\binom{T+1}{s+1}.
\]
Finally, we use Pascal’s rule $\binom{T+1}{s+1}=\binom{T}{s+1}+\binom{T}{s}$ to rewrite
\[
(N-m)\binom{T}{s}-s\binom{T+1}{s+1}
=\big[N-m-s\big]\binom{T}{s}-s\binom{T}{s+1}.
\]
Recalling $T=N-1-m$, we obtain the \emph{exact} closed form
\begin{equation}
    \boxed{\;
    g_m
    =(N-m)\ +\!\!\sum_{s=1}^{m}
    \binom{m}{s}\rho^{s}\Big[(N-m-s)\binom{N-1-m}{s}
    \ -\ s\binom{N-1-m}{s+1}\Big] = \sum_{s=0}^{m}
    \binom{m}{s}\rho^{s}\binom{N-m}{s+1}.
    \;}
    \label{eq:gm-exact}
\end{equation}

\subsection*{Large-$N$ asymptotics and the critical gap scaling}
We identify the critical transition point by the condition that $S_{N-1}\approx1$: \textit{i.e.} when the total relative partition function of all length-($N-1$) chains is comparable with the relative weight of the HP. 
Recall $\gamma$ is the energy penalty of a weak bond: $\gamma=e^{-\Delta E}=\mathrm{e}^{-b}N^{-a}$.
From \eqref{eq:comp-form} with $k=N-1$, the $p=1$ term is
\[
T_1=(1-\gamma)^{N-2}\,g_{N-1}.
\]
From above, we have $g_{N-1}=1$ and $T_1=(1-\gamma)^{N-2}=1-(N-2)\gamma+O(\gamma^2N^2)=1+o(1)$ provided $a\geq2$. The leading correction comes from $p=2$:
\begin{equation}
    T_2=\gamma(1-\gamma)^{N-3}\sum_{m=1}^{N-2}g_m\,g_{N-1-m}
    \label{eq:T2-start}
\end{equation}

\paragraph{Saddle–point asymptotics.}
Starting from \eqref{eq:gm-exact}, write:
\begin{equation}
\label{eq:qms}
g_m = \sum_{s=0}^m q_{m,s} \quad , \quad q_{m,s}=\binom{m}{s}\rho^{s}\binom{N-m}{s+1}\end{equation}
Then, starting from the exact two–block term \eqref{eq:T2-start}, we have:
\begin{equation}
\label{eq:T2qms}
T_2=\gamma(1-\gamma)^{N-3}\sum_{m=1}^{N-2} \left( \sum_{s_1=0}^m q_{m,s_1} \right)\left( \sum_{s_2=0}^{N-1-m} q_{N-1-m,s_2} \right)
\end{equation}

We will analyze the large–$N$ behavior by approximating these sums with a Laplace (saddle–point) method over $m$, $s_1$, and $s_2$. To do this, in the large $N$ limit we set $m=xN$, $s_1=\alpha_1 N$, and $s_2=\alpha_2 N$, so that the $q_{m,s}$ terms are replaced by functions:
\[q_{m,s_1}\to q(xN, \alpha_1 N) \quad , \quad 
q_{N-1-m,s_2}\to q\left((1-x)N, \alpha_2 N\right)
\]
where $x\in(0,1)$, $\alpha_1\in[0,x]$, and $\alpha_2\in [0, 1-x]$. We also apply Stirling’s approximation to the binomials in \eqref{eq:qms}, obtaining
\[
\log\left[ q(xN, \alpha N)\right]=N\,\psi(x,\alpha;\rho)+o(N),
\quad
\psi(x,\alpha;\rho):=x\,H\!\Big(\tfrac{\alpha}{x}\Big)+(1-x)\,H\!\Big(\tfrac{\alpha}{1-x}\Big)+\alpha\log\rho,
\]
where $H(u):=-u\log u-(1-u)\log(1-u)$ is the binary entropy. Hence
\[
\log\left[q(xN, \alpha_1 N)q((1-x)N, \alpha_2 N)\right] = N\big[\psi(x,\alpha_1;\rho)+\psi(1-x,\alpha_2;\rho)\big]+o(N).
\]
Then, replacing the sums in \eqref{eq:T2qms} by integrals, we obtain
\begin{equation}
\label{eq:T2qms_int}
T_2\approx \gamma(1-\gamma)^{N-3}\int_0^1 dx\int_0^x d\alpha_1 \exp\left(N\big[\psi(x,\alpha_1;\rho)\big]\right)\int_0^{1-x}d\alpha_2 \exp\left(N\big[\psi(1-x,\alpha_2;\rho)\big]\right)
\end{equation}
Now we proceed with successive Laplace approximations in $\alpha_1$, $\alpha_2$, and $x$. By symmetry in $x$, the maximizer is at $x^\star=\tfrac12$, and the maximizers satisfy $\alpha_1^* = \alpha_2^*$. 
Applying Laplace approximations, we obtain
\begin{equation}
\label{eq:T2qms_int_laplace}
T_2\approx \gamma(1-\gamma)^{N-3} \left(\frac{2\pi}{N}\right)^{3/2} C(\rho) e^{\Lambda(\rho)N},
\end{equation}
where $C(\rho)>0$ is a function collecting the Laplace approximation prefactors, and we have defined the exponential rate
\[
\Lambda(\rho)= 2\max_{0\le \alpha\le 1/2} \big\{ \psi(\tfrac{1}{2},\alpha;\rho)\big\} = 2\log\!\big(1+\sqrt{\rho}\,\big).
\]
Therefore, with $\gamma=e^{-\Delta E}=e^{-b}N^{-a}$, we obtain the large-$N$ scaling
\[
T_2\sim \gamma\,C(\rho)\,N^{-3/2}\,e^{\Lambda(\rho)N},
\]
and the critical condition $T_2=O(1)$ gives
\[\Delta E_c \sim N \Lambda(\rho) - \frac{3}{2}\log N + O(1)\]
Noting $\Delta E = a\log N + b$, the critical gap-scaling coefficient is
\[
\boxed{\,a(N)=\Lambda(\rho) \frac{N}{\log N} -\frac{3}{2}-\frac{b}{\log N}+o(1),
\qquad \Lambda(\rho)=2\log(1+\sqrt{\rho})\, .}
\]

\subsection*{Remarks}
The above argument assumed equal concentrations of components. Using unequal monomer concentrations will simply reweigh the entries of the probability matrix \eqref{eq:Amatrix}, so that $A$ is replaced by its element-wise product with the concentration matrix, \(C=(c_{ij})\): \(A\to \widetilde A=A\odot C\). The large-\(N\) limit for \(S_{N-1}\) given in \eqref{eq:openchain-matrix} depends on the scaling of the row sums \((\widetilde A\mathbf{1})_i\sim \rho\sum_{j>i+1}c_{ij}\). If the row-averaged concentrations approach a positive constant, \(\overline c>0\), then all combinatorics go through with an \emph{effective} shortcut density \(\rho_{\mathrm{eff}}=\rho\,\overline{c}\), so the dominant entropy factor \(e^{\Lambda(\rho) N}\) becomes \(e^{\Lambda(\rho_{\mathrm{eff}}) N}\) and the critical scaling reads \(a(N)\approx \tfrac{\Lambda(\rho_{\mathrm{eff}})}{\log N}N-\tfrac{3}{2}-\tfrac{b}{\log N}\). Heterogeneity across rows only shifts \(\rho\) within bounds set by the min/max row-average of \(c_{ij}\) and thus brackets the same \(N\) scaling law. Therefore the scaling laws applied to the unequal concentration case should still hold for some effective $\rho_{\rm eff}$.

\section{Relationship between alphabet size and expected number of overlaps}

In this section, we derive an analytic method to find the configuration of $n$ strings that minimizes the average number of pairwise overlaps between strings (\textit{i.e.} minimizing cases when any two strings have the same symbol at any given position). The strings are of length $L$ with entries taken from an alphabet of size $A$.

Let $n$ denote the number of distinct strings (\textit{i.e.} distinct interactions that must be designed), and let $\{x^{(1)}, x^{(2)}, \ldots, x^{(n)}\}$ denote the $n$ strings, where each $x^{(i)}$ is an $L$-component vector with components $x^{(i)}_\ell \in \{1, 2, \ldots, A\}$. 

The number of overlaps between strings $x^{(i)}$ and $x^{(j)}$ is
\begin{equation}
O(x^{(i)}, x^{(j)}) = \sum_{\ell=1}^L \mathbf{1}\left(x_\ell^{(i)} = x_\ell^{(j)}\right),
\end{equation}
where $\mathbf{1}$ is an indicator function. Then, the average number of pairwise overlaps across all $n$ strings is:
\begin{equation}
\bar{n}_w = \frac{1}{\binom{n}{2}} \sum_{1 \leq i < j \leq n} O(x^{(i)}, x^{(j)}) = \frac{2}{n(n-1)} \sum_{1 \leq i < j \leq n} O(x^{(i)}, x^{(j)})
\end{equation}

\begin{theorem}
The average number of pairwise overlaps for $n$ strings of length $L$ over an alphabet of size $A$ is:
\begin{equation}
\bar{n}_w= \frac{L}{n(n-1)} \sum_{\ell=1}^L \left[ \sum_{a=1}^A k_{\ell,a}^2 - n \right]
\end{equation}
where $k_{\ell,a}$ denotes the number of strings with the ``letter'' $a$ at position $\ell$. 
\end{theorem}

\begin{proof}
Let $k_{\ell, a}$ denote the number of strings with symbol $a$ in position $\ell$. Then, we must have $\sum_{a=1}^A k_{\ell,a} = n$.

For a fixed $\ell$, the number of pairs between the $n$ strings which overlap at position $\ell$ is:
\begin{equation}
\sum_{a=1}^A \binom{k_{\ell,a}}{2} = \frac{1}{2} \sum_{a=1}^A k_{\ell,a}(k_{\ell,a} - 1) = \frac{1}{2} \left( \sum_{a=1}^A k_{\ell,a}^2 - \sum_{a=1}^A k_{\ell,a} \right) = \frac{1}{2} \left( \sum_{a=1}^A k_{\ell,a}^2 - n \right)
\end{equation}
where we have assumed all $k_{\ell, a}>1$.  
The total number of overlaps across all pairs and all positions is:
\begin{equation}
\sum_{1 \leq i < j \leq n} O(x^{(i)}, x^{(j)}) =  \sum_{\ell=1}^L \sum_{a=1}^A \binom{k_{\ell,a}}{2} =  \sum_{\ell=1}^L \frac{1}{2} \left( \sum_{a=1}^A k_{\ell,a}^2 - n \right)
\end{equation}
Thus, the average number of pairwise overlaps is:
\begin{equation}\label{eq:nw_expression}
\bar{n}_w = \frac{2}{n(n-1)} \sum_{\ell=1}^L \frac{1}{2} \left( \sum_{a=1}^A k_{\ell,a}^2 - n \right) = \frac{1}{n(n-1)} \left[ \sum_{\ell=1}^L \sum_{a=1}^A k_{\ell,a}^2 - nL \right]
\end{equation}
\end{proof}
To minimize $\bar{n}_w$, we must minimize $\sum_{a=1}^A k_{\ell,a}^2$ for each position $\ell$.

\begin{lemma}
    For fixed $n$ and $A$, the sum $\sum_{a=1}^A k_{\ell,a}^2$ satisfying $\sum_{a=1}^A k_{\ell,a} = n$ is minimized when the values of $k_{\ell,a}$ are as equal as possible. The minimum value is

\begin{equation}
\left[\sum_{\alpha=1}^A k_{\ell,\alpha}^2\right]_{\min} = (A-r)q^2 + r(q+1)^2
\end{equation}
    
\end{lemma}

\begin{proof}
Since $\sum_{a=1}^A k_{\ell,a} = n$, and the sum $\sum_{a=1}^A k_{\ell,a}^2$ is a convex function of the variables $k_{\ell,a}$, the sum will be minimized when the values of $k_{\ell,a}$ are as equal as possible.

Let $q = \lfloor n/A \rfloor$ (floor division by $A$), and $r = n \bmod A$, so that $n = Aq + r $. Then, the optimal choice of $n$ symbols for the $\ell$th position has
\begin{itemize}
\item $(A-r)$ symbols appearing $q$ times each
\item $r$ symbols appearing $(q+1)$ times each
\end{itemize}

This gives the minimum value:
\begin{equation}
    \left[\sum_{\alpha=1}^A k_{\ell,\alpha}^2\right]_{\min} = (A-r)q^2 + r(q+1)^2
\end{equation}

\end{proof}

\begin{corollary}
The minimum average number of pairwise overlaps is:

\textbf{Case 1:} If $A$ divides $n$, then $q = n/A$ and $r = 0$:
$$\left[\bar{n}_w\right]_{\min} = \frac{L}{n(n-1)} \left[ A \cdot \left(\frac{n}{A}\right)^2 - n \right] = \frac{L}{n(n-1)} \left[ \frac{n^2}{A} - n \right] = L \cdot \frac{1}{A} \cdot \frac{n}{n-1}$$

\textbf{Case 2:} General case with $q = \lfloor n/A \rfloor$ and $r = n \bmod A$:
\begin{equation}
\label{eq:nw_expression_final}
    \left[\bar{n}_w\right]_{\min} = \frac{L}{n(n-1)} \left[ (A-r)q^2 + r(q+1)^2 - n \right]
\end{equation}
\end{corollary}

\section{Phase Diagrams for Additional Parameters}
Fig.~\ref{fig:supp_dna} and Fig.~\ref{fig:supp_magnets} provide phase diagrams for additional $b$ and $\rho_b$ values, sweeping $b=\{0, 2\}$ and $\rho_b = \{0.05, 0.1, 0.3\}$. Intersection points between the design and HPP solving curves are noted, showing the largest possible $N$ for each system. If the intersection point occurs for $a<2$ (the energy bound found in \cite{murugan2015undesired} corresponding to a trivial graph with $\rho_b=0$), the intersection point is not labeled. The bound $a\geq2$ is a reasonable constraint on energy versus entropy, ensuring that the energy penalty for replacing a strong bond with a weak bond, $e^{\beta(w - s)}$, outweighs the number of ways to form that weak bond $\sim N^2$. Table \ref{tab:HPPexamples} lists the design parameters and maximum HPP graph size $N$ which can be solved using our methods for several systems and varying $\rho_b$ and $b$ values. In particular, we include the magnetic dipole system, heuristic DNA model, as well as the lock/key colloid ``Pacman'' particles realized in \cite{sacanna_lock_2010} and modeled for their information capacity in \cite{huntley_infocap}. For the ``Pacman'' particles, in \cite{huntley_infocap} it was found that an alphabet of 6 distinct particles could be engineered, and we estimate the average $\epsilon_w/\epsilon_s$ of their system to be 0.2.

\begin{figure*}[ht]
\centering
\begin{subfigure}{0.49\textwidth}
\includegraphics[width=\textwidth]{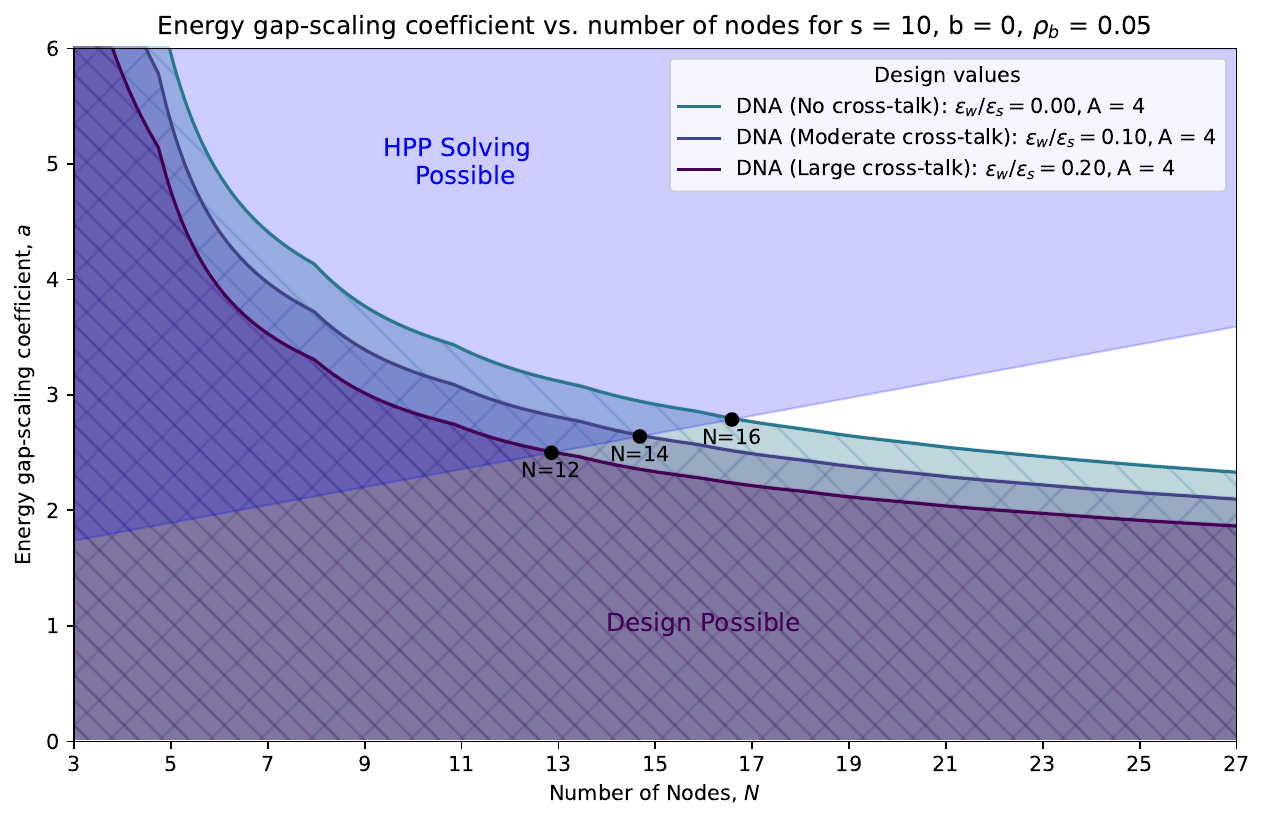}
    \caption{\label{}}
\end{subfigure}
\begin{subfigure}{0.49\textwidth}
\includegraphics[width=\textwidth]{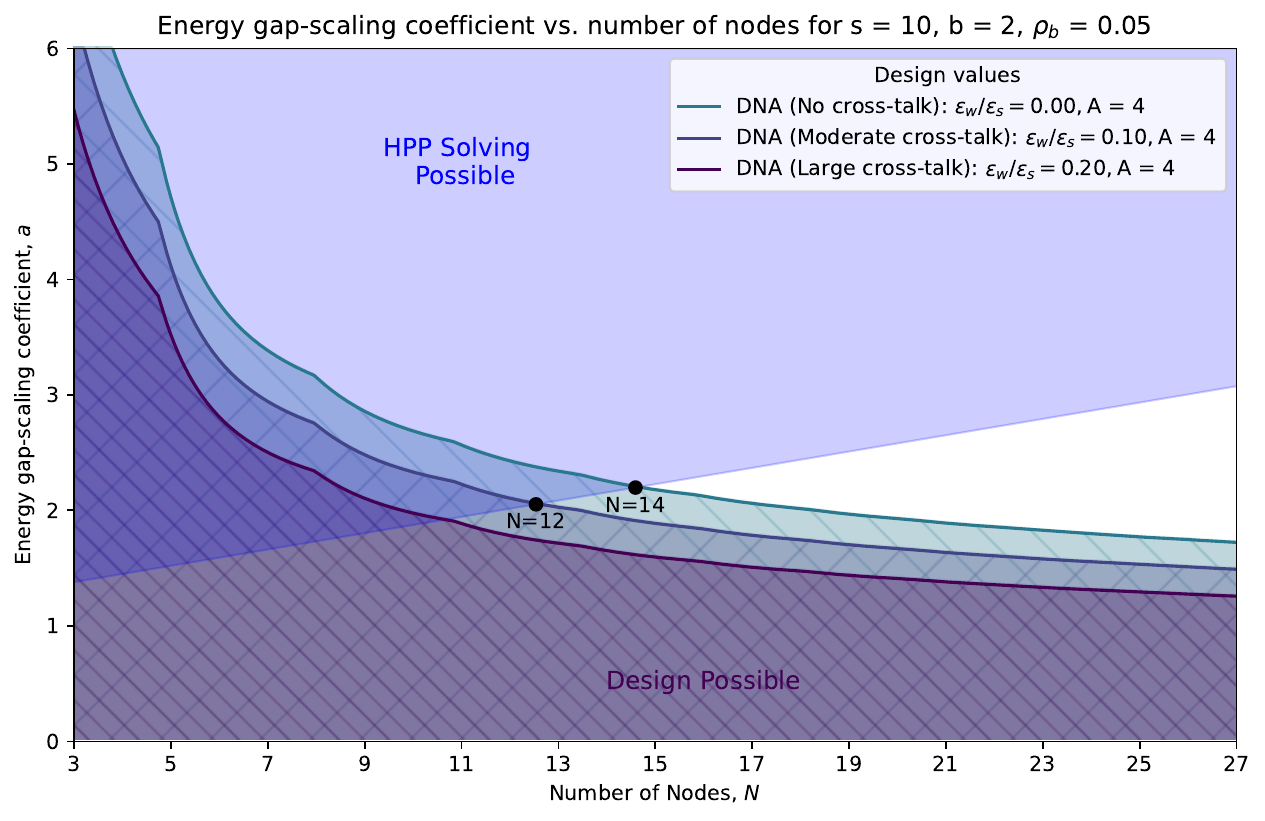}
    \caption{\label{}}
\end{subfigure}
\begin{subfigure}{0.49\textwidth}
    \includegraphics[width=\textwidth]{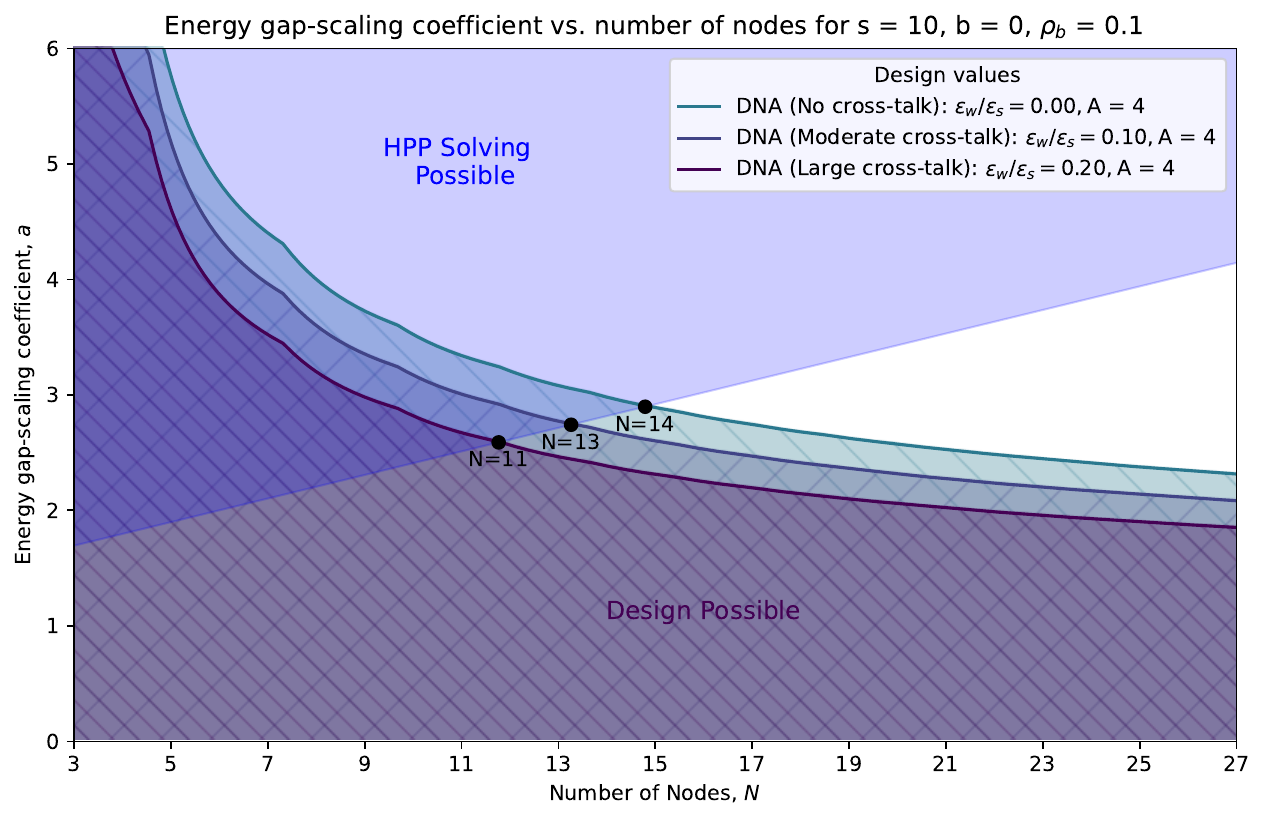}
    \caption{\label{}}
\end{subfigure}
\begin{subfigure}{0.49\textwidth}
    \includegraphics[width=\textwidth]{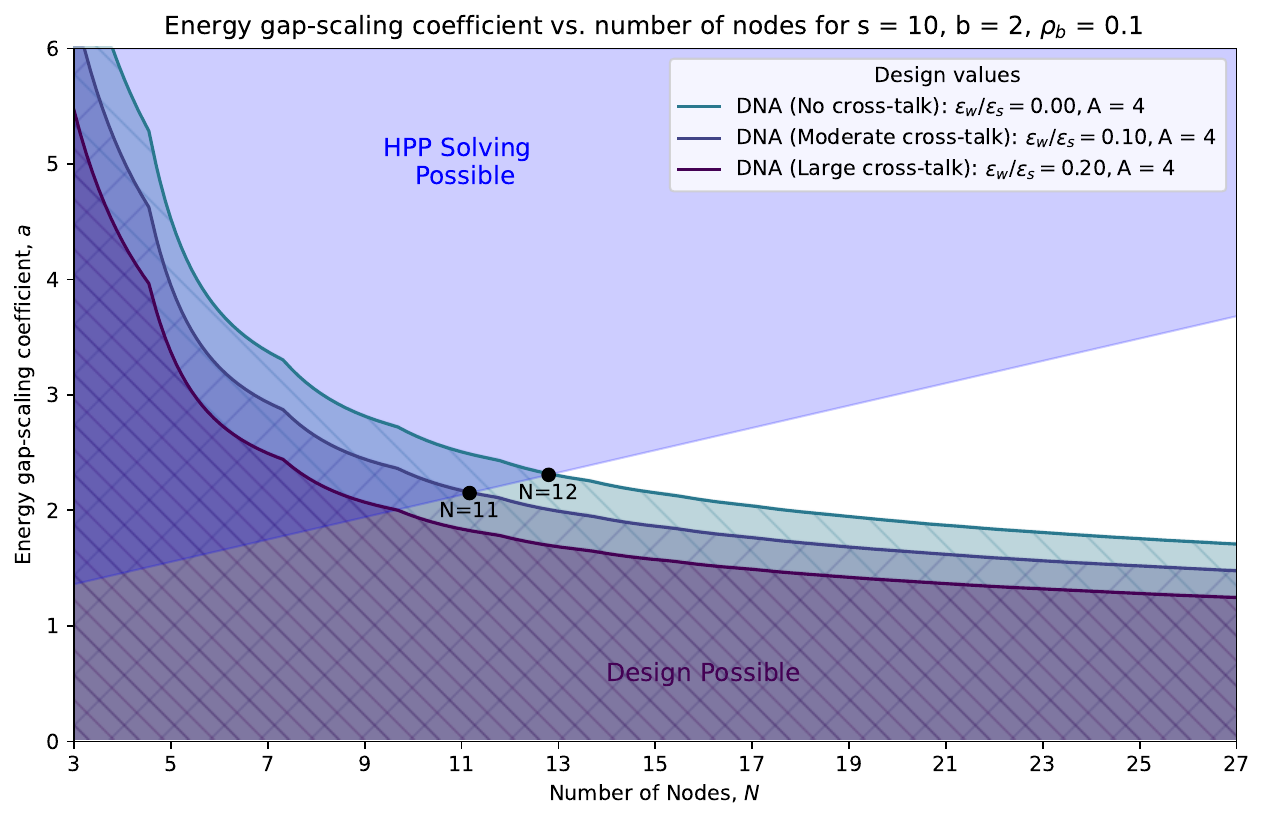}
    \caption{\label{}}
\end{subfigure}
\begin{subfigure}{0.49\textwidth}
    \includegraphics[width=\textwidth]{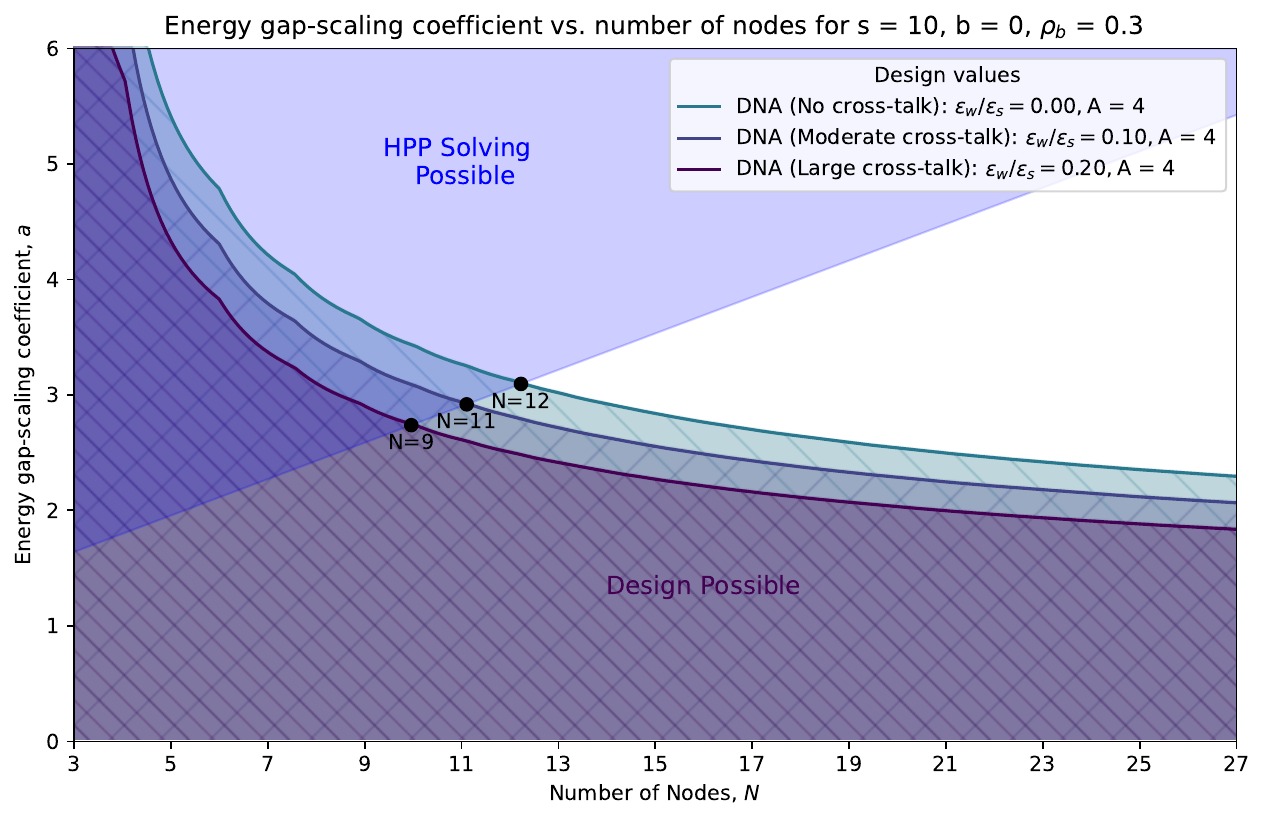}
    \caption{\label{}}
\end{subfigure}
\begin{subfigure}{0.49\textwidth}
    \includegraphics[width=\textwidth]{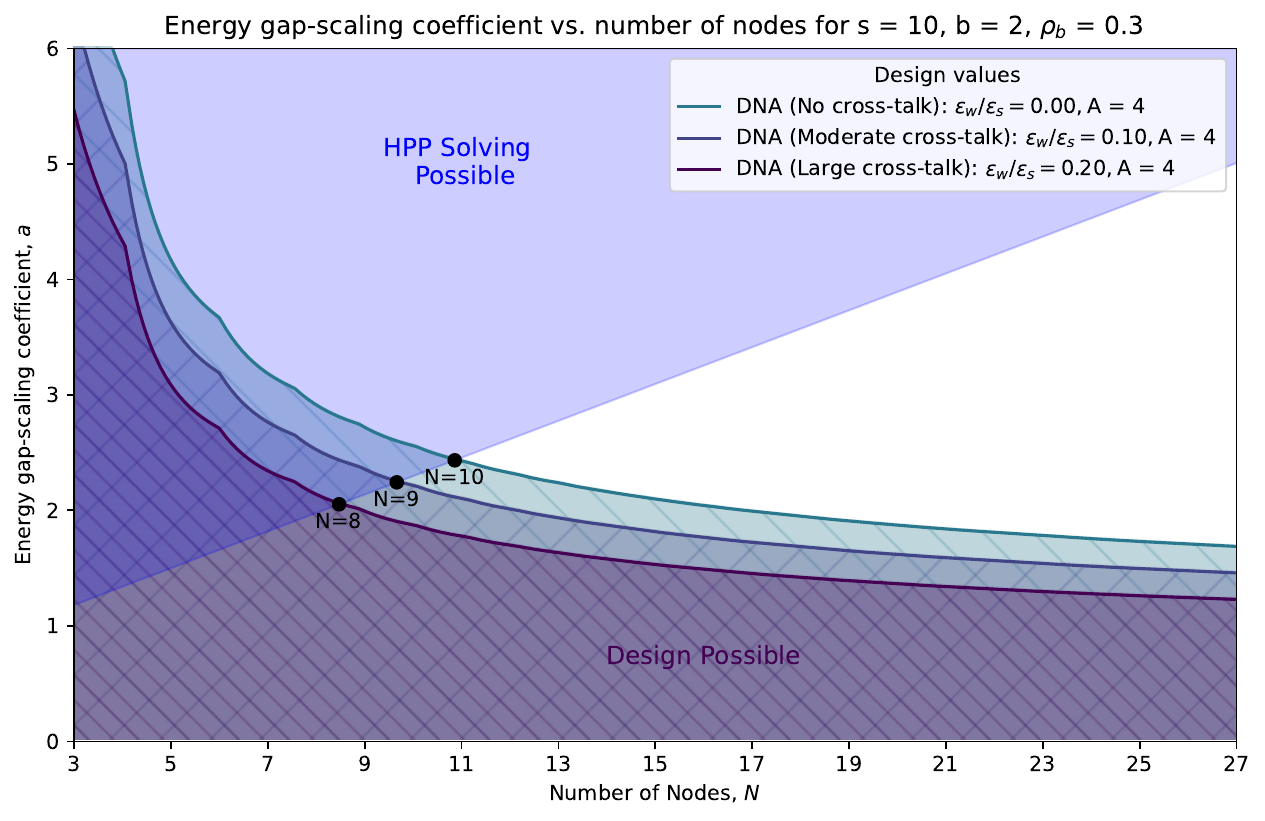}
    \caption{\label{}}
\end{subfigure}
\caption{\label{fig:supp_dna}\textbf{Phase Diagrams for DNA-like system for several $b$ and $\rho_b$ values.} The kinks visible in the design curves, especially for small $N$, are due to the presence of floor division by the alphabet size, $A=4$ in the equation for $\frac{\bar{n}_w}{L}$ in \eqref{eq:nw_expression_final}. Design curves are shown for three example DNA-like systems: (i) No crosstalk ($A=4$, $\epsilon_w/\epsilon_s=0$) (ii) Moderate crosstalk ($A=4$, $\epsilon_w/\epsilon_s=0.1$) (iii) Large crosstalk ($A=4$, $\epsilon_w/\epsilon_s$).}
\end{figure*}


\begin{figure*}[ht]
\centering
\begin{subfigure}{0.49\textwidth}
\includegraphics[width=\textwidth]{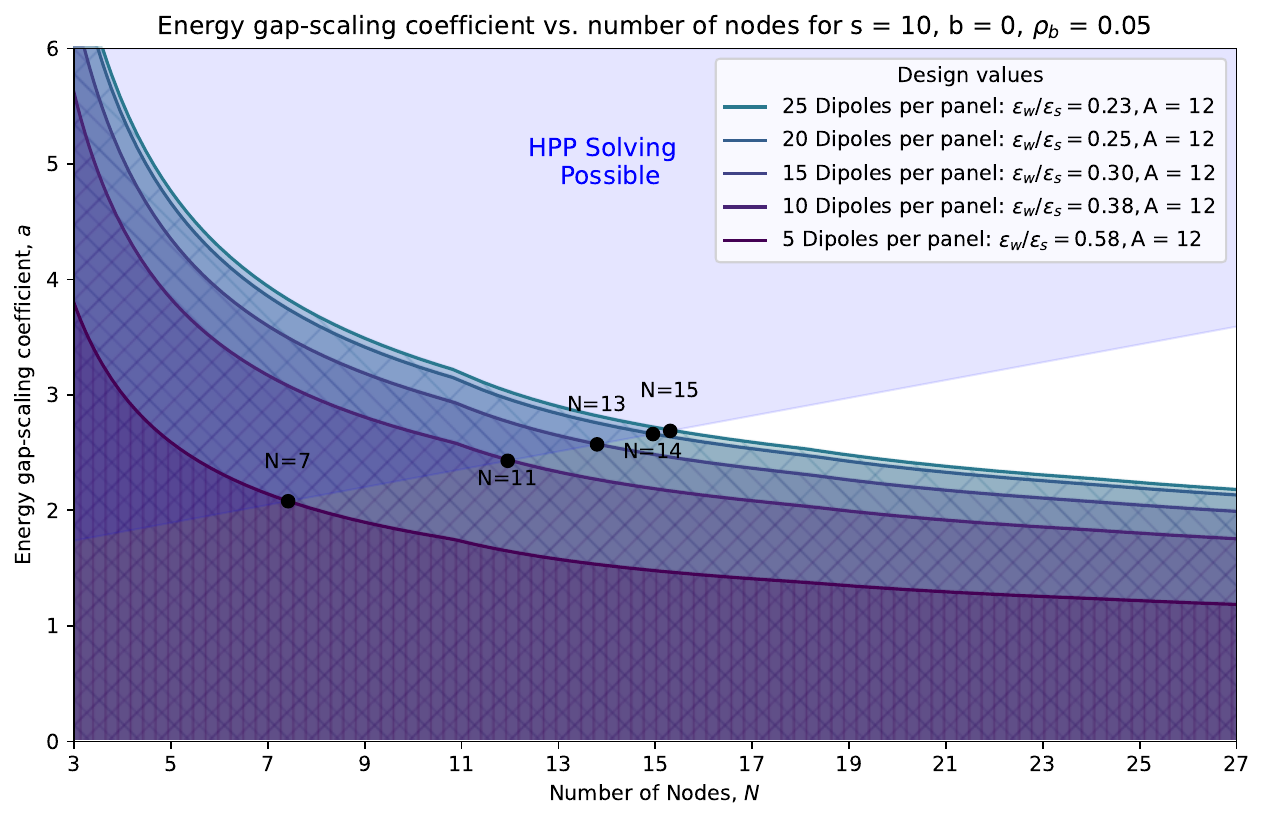}
    \caption{\label{}}
\end{subfigure}
\begin{subfigure}{0.49\textwidth}
\includegraphics[width=\textwidth]{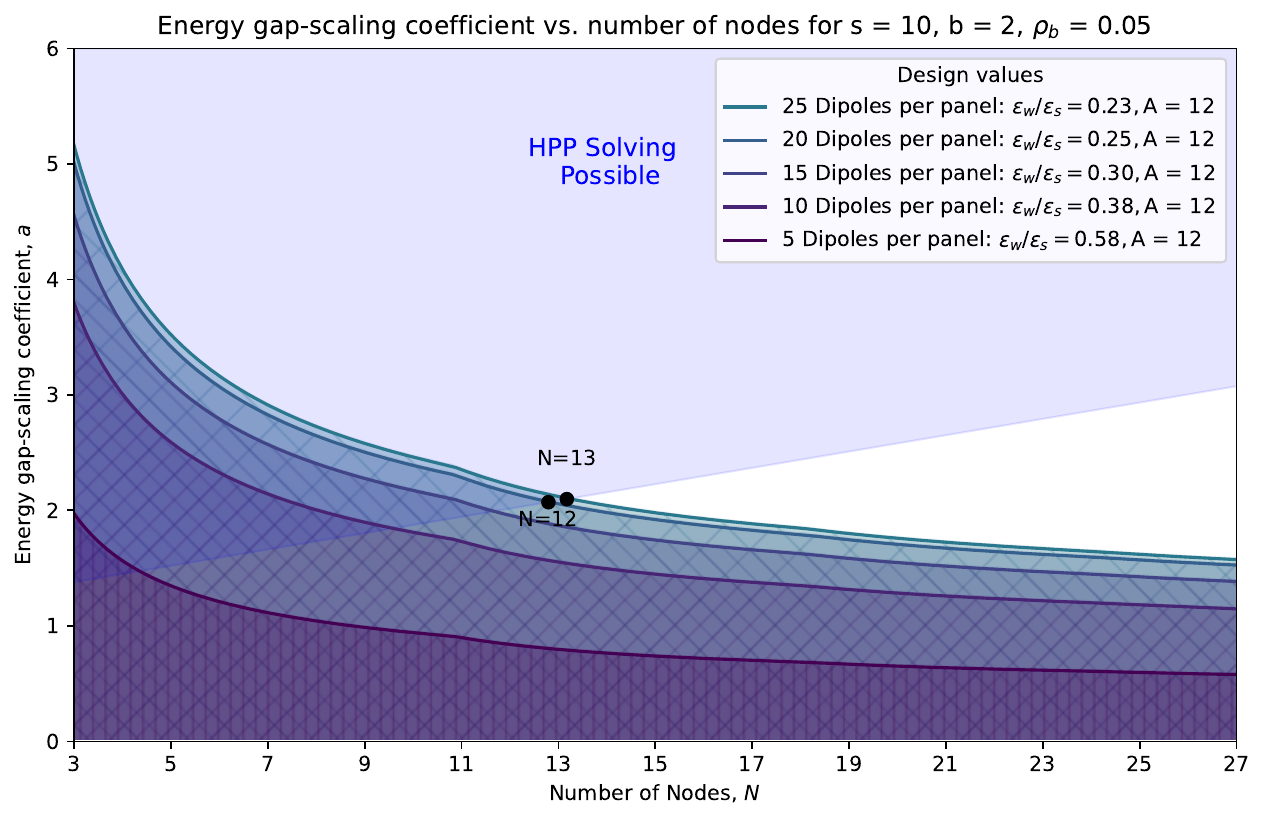}
    \caption{\label{}}
\end{subfigure}
\begin{subfigure}{0.49\textwidth}
    \includegraphics[width=\textwidth]{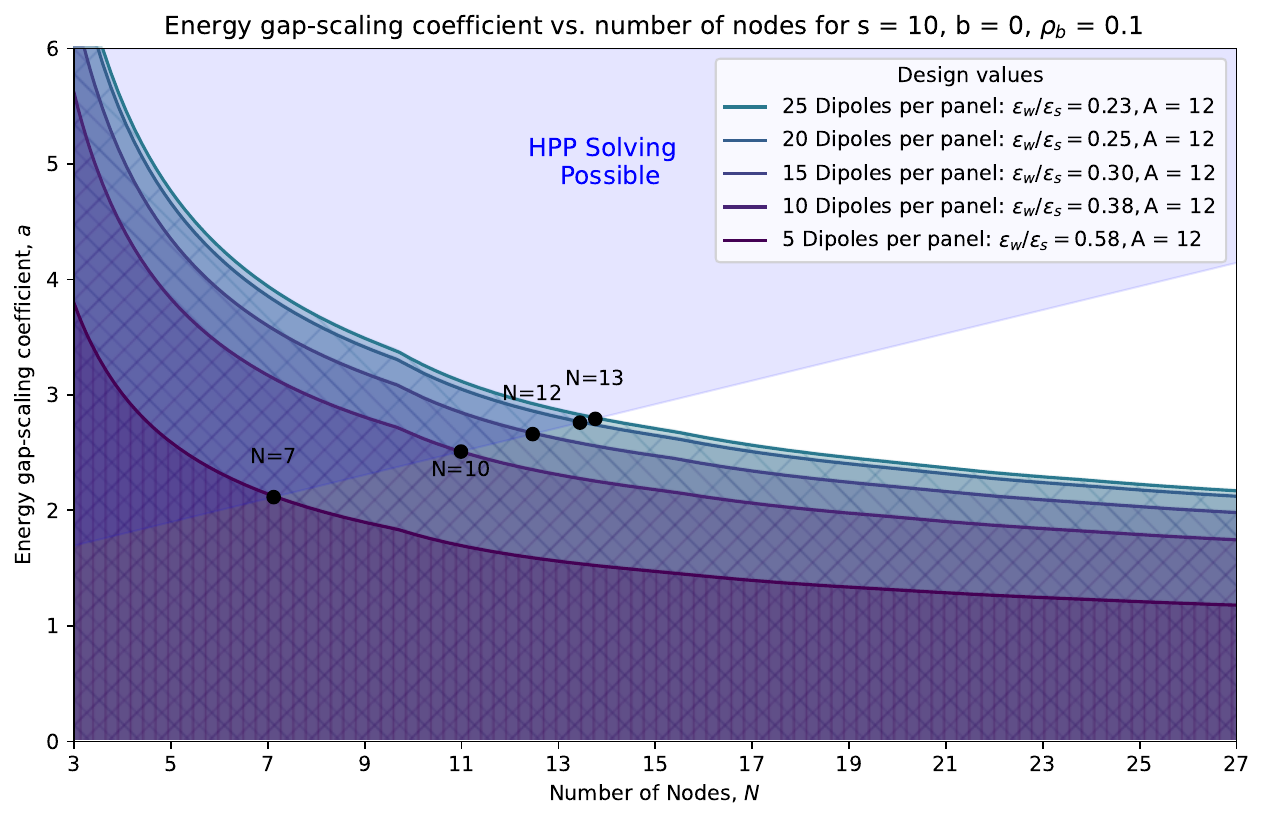}
    \caption{\label{}}
\end{subfigure}
\begin{subfigure}{0.49\textwidth}
    \includegraphics[width=\textwidth]{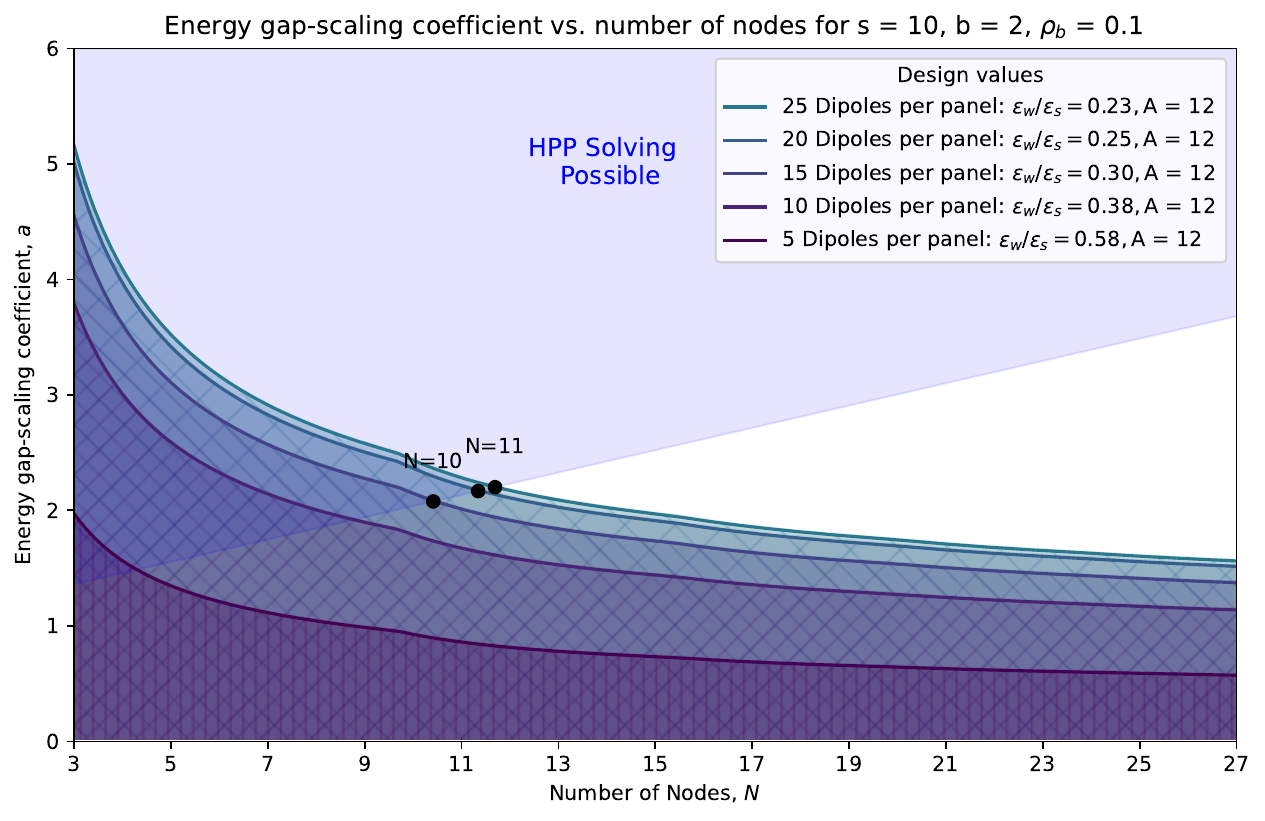}
    \caption{\label{}}
\end{subfigure}
\begin{subfigure}{0.49\textwidth}
    \includegraphics[width=\textwidth]{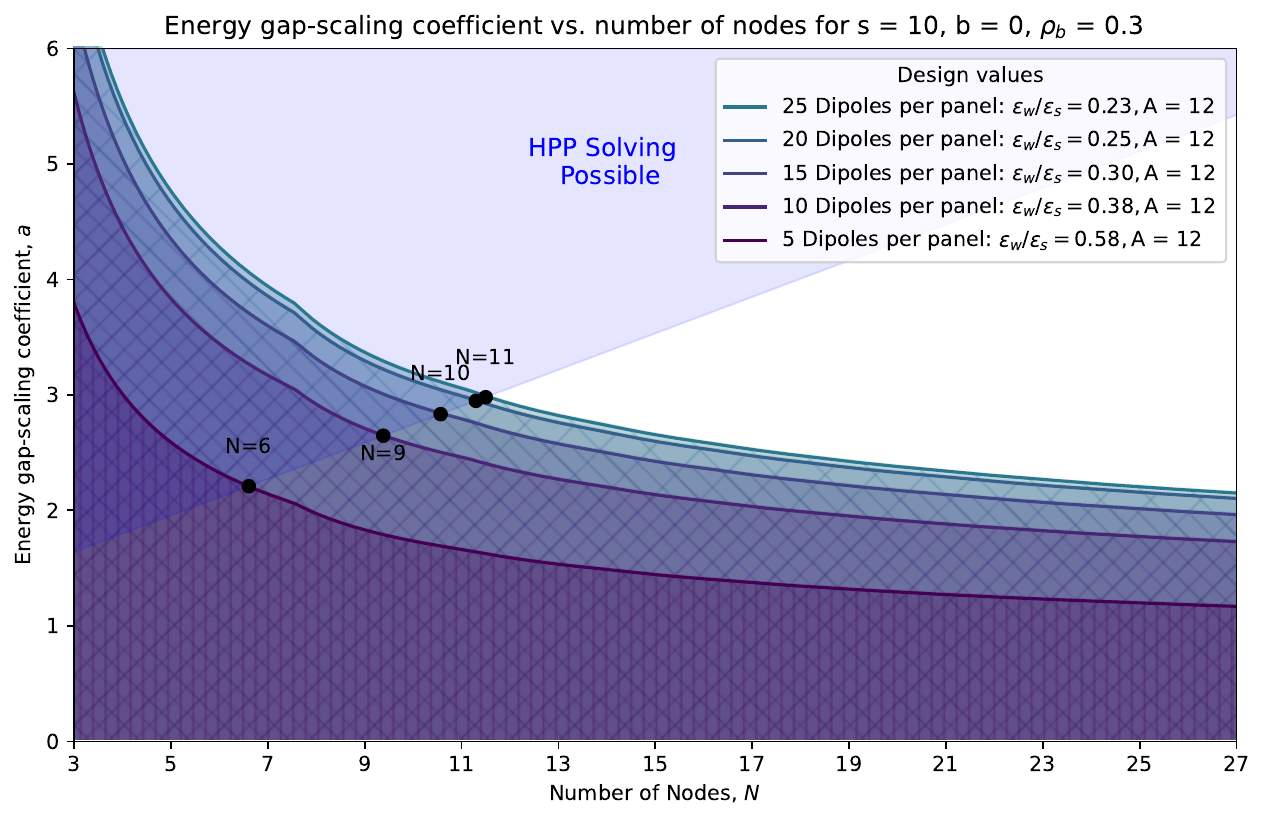}
    \caption{\label{}}
\end{subfigure}
\begin{subfigure}{0.49\textwidth}
    \includegraphics[width=\textwidth]{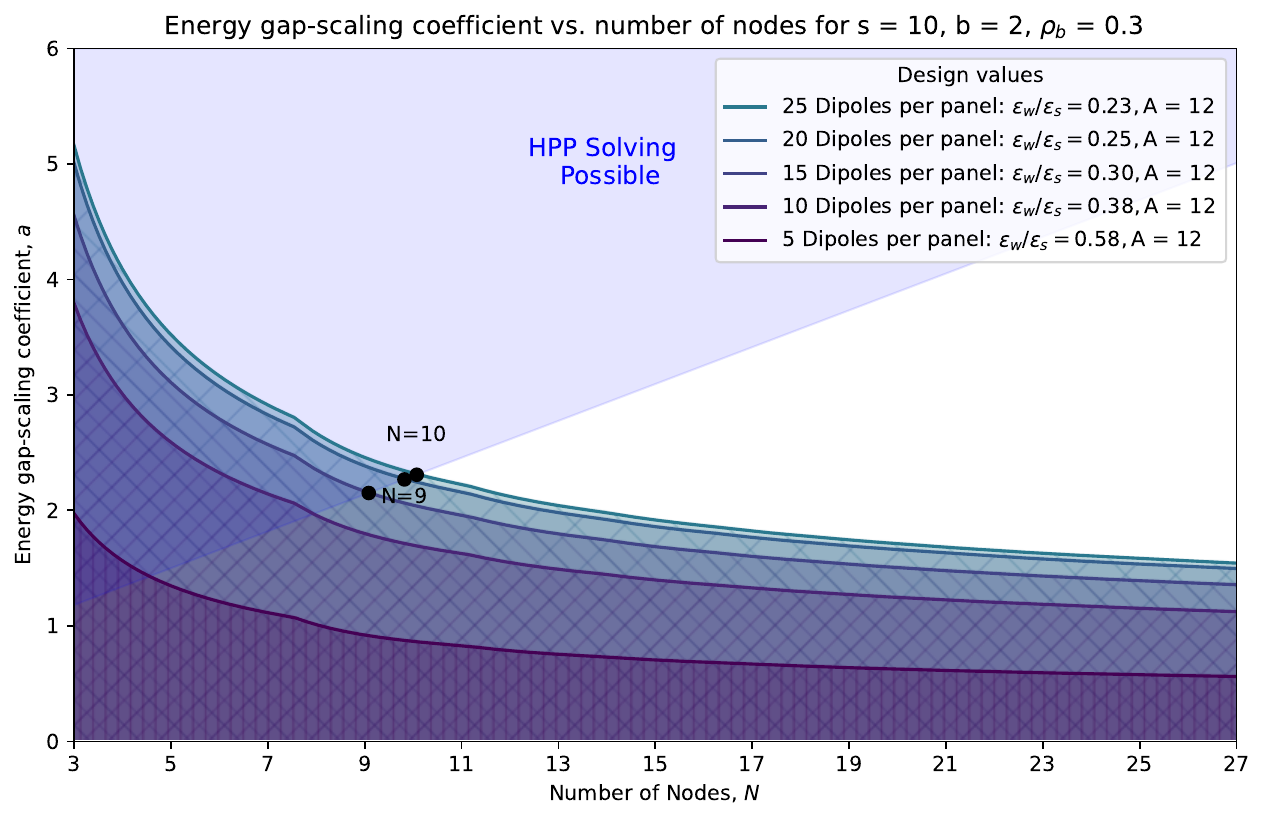}
    \caption{\label{}}
\end{subfigure}
\caption{\label{fig:supp_magnets}\textbf{Phase Diagrams for magnets systems for several $b$ and $\rho_b$ values.} Design curves are shown for five of the magnet systems detailed in \cite{du2022programming}, all with $A=12$ and with $\epsilon_w/\epsilon_s$ values ranging from 0.58 to 0.23.} 

\end{figure*}

    
    
    
    

\begin{table*}[ht]
  \centering
  \caption{\label{tab:HPPexamples} The design parameters and the maximum problem size $N$ that can be attained HPP-solving parameters for several example systems.}
  \begin{tabular}{@{} l c c c c c c @{}}  
    \toprule
    \multicolumn{1}{c}{\textbf{System}} &
    \multicolumn{3}{c}{\textbf{Design}} &
    \multicolumn{3}{c}{\textbf{HPP}}\\
    \cmidrule(r){2-4}\cmidrule(l){5-7}
    & $A$ &
      Mean $\epsilon_w / \epsilon_s$ &
      $ \lvert s \rvert \ [k_B T]$ &
      $b$ &
      $\rho_b$ &
      Max.\ $N$ \\
    \midrule

Ideal, best N & 1000 & 0.00 & 10 & 0 & 0.05 & 22 \\
Ideal, best N & 1000 & 0.00 & 10 & 0 & 0.1 & 19 \\
Ideal, best N & 1000 & 0.00 & 10 & 0 & 0.3 & 15 \\
25-Dipole Magnet Panels & 12 & 0.23 & 10 & 0 & 0.05 & 15 \\
25-Dipole Magnet Panels & 12 & 0.23 & 10 & 2 & 0.05 & 13 \\
25-Dipole Magnet Panels & 12 & 0.23 & 10 & 0 & 0.1 & 13 \\
25-Dipole Magnet Panels & 12 & 0.23 & 10 & 2 & 0.1 & 11 \\
25-Dipole Magnet Panels & 12 & 0.23 & 10 & 0 & 0.3 & 11 \\
25-Dipole Magnet Panels & 12 & 0.23 & 10 & 2 & 0.3 & 10 \\
20-Dipole Magnet Panels & 12 & 0.25 & 10 & 0 & 0.05 & 14 \\
20-Dipole Magnet Panels & 12 & 0.25 & 10 & 2 & 0.05 & 12 \\
20-Dipole Magnet Panels & 12 & 0.25 & 10 & 0 & 0.1 & 13 \\
20-Dipole Magnet Panels & 12 & 0.25 & 10 & 2 & 0.1 & 11 \\
20-Dipole Magnet Panels & 12 & 0.25 & 10 & 0 & 0.3 & 11 \\
20-Dipole Magnet Panels & 12 & 0.25 & 10 & 2 & 0.3 & 9 \\
15-Dipole Magnet Panels & 12 & 0.30 & 10 & 0 & 0.05 & 13 \\
15-Dipole Magnet Panels & 12 & 0.30 & 10 & 2 & 0.05 & 11 \\
15-Dipole Magnet Panels & 12 & 0.30 & 10 & 0 & 0.1 & 12 \\
15-Dipole Magnet Panels & 12 & 0.30 & 10 & 2 & 0.1 & 10 \\
15-Dipole Magnet Panels & 12 & 0.30 & 10 & 0 & 0.3 & 10 \\
15-Dipole Magnet Panels & 12 & 0.30 & 10 & 2 & 0.3 & 9 \\
10-Dipole Magnet Panels & 12 & 0.38 & 10 & 0 & 0.05 & 11 \\
10-Dipole Magnet Panels & 12 & 0.38 & 10 & 0 & 0.1 & 10 \\
10-Dipole Magnet Panels & 12 & 0.38 & 10 & 2 & 0.1 & 8 \\
10-Dipole Magnet Panels & 12 & 0.38 & 10 & 0 & 0.3 & 9 \\
10-Dipole Magnet Panels & 12 & 0.38 & 10 & 2 & 0.3 & 7 \\
5-Dipole Magnet Panels & 12 & 0.58 & 10 & 0 & 0.05 & 7 \\
5-Dipole Magnet Panels & 12 & 0.58 & 10 & 0 & 0.1 & 7 \\
5-Dipole Magnet Panels & 12 & 0.58 & 10 & 0 & 0.3 & 6 \\
DNA (No cross-talk) & 4 & 0.00 & 10 & 0 & 0.05 & 16 \\
DNA (No cross-talk) & 4 & 0.00 & 10 & 2 & 0.05 & 14 \\
DNA (No cross-talk) & 4 & 0.00 & 10 & 0 & 0.1 & 14 \\
DNA (No cross-talk) & 4 & 0.00 & 10 & 2 & 0.1 & 12 \\
DNA (No cross-talk) & 4 & 0.00 & 10 & 0 & 0.3 & 12 \\
DNA (No cross-talk) & 4 & 0.00 & 10 & 2 & 0.3 & 10 \\
DNA (Moderate cross-talk) & 4 & 0.10 & 10 & 0 & 0.05 & 14 \\
DNA (Moderate cross-talk) & 4 & 0.10 & 10 & 2 & 0.05 & 12 \\
DNA (Moderate cross-talk) & 4 & 0.10 & 10 & 0 & 0.1 & 13 \\
DNA (Moderate cross-talk) & 4 & 0.10 & 10 & 2 & 0.1 & 11 \\
DNA (Moderate cross-talk) & 4 & 0.10 & 10 & 0 & 0.3 & 11 \\
DNA (Moderate cross-talk) & 4 & 0.10 & 10 & 2 & 0.3 & 9 \\
DNA (Large cross-talk) & 4 & 0.20 & 10 & 0 & 0.05 & 12 \\
DNA (Large cross-talk) & 4 & 0.20 & 10 & 2 & 0.05 & 10 \\
DNA (Large cross-talk) & 4 & 0.20 & 10 & 0 & 0.1 & 11 \\
DNA (Large cross-talk) & 4 & 0.20 & 10 & 2 & 0.1 & 9 \\
DNA (Large cross-talk) & 4 & 0.20 & 10 & 0 & 0.3 & 9 \\
DNA (Large cross-talk) & 4 & 0.20 & 10 & 2 & 0.3 & 8 \\
Pac man particles & 6 & 0.20 & 10 & 0 & 0.05 & 14 \\
Pac man particles & 6 & 0.20 & 10 & 2 & 0.05 & 12 \\
Pac man particles & 6 & 0.20 & 10 & 0 & 0.1 & 13 \\
Pac man particles & 6 & 0.20 & 10 & 2 & 0.1 & 10 \\
Pac man particles & 6 & 0.20 & 10 & 0 & 0.3 & 10 \\
Pac man particles & 6 & 0.20 & 10 & 2 & 0.3 & 9 \\
    \bottomrule
  \end{tabular}
\end{table*}

\bibliography{bib}